\documentclass[11pt]{article}

\usepackage{amssymb, amsmath, verbatim, amsthm,url, multirow,fullpage,mathtools}
\usepackage{longtable, rotating,makecell,array}
\usepackage[aligntableaux=top]{ytableau}

\usepackage[colorinlistoftodos,textsize=footnotesize]{todonotes}
\usepackage[all]{xy}
\xyoption{poly}
\xyoption{arc}
\usetikzlibrary{decorations.pathmorphing,calc}
\usetikzlibrary{intersections}

\definecolor{light-gray}{gray}{0.6}

\tikzstyle{propagator}=[decorate,decoration={snake,amplitude=0.8mm}]
\tikzstyle{smallpropagator}=[decorate,decoration={snake,segment length=3mm,amplitude=0.5mm}]

\tikzstyle{firstdash}=[dashed,line cap=round, dash pattern=on 2pt off 1pt]
\tikzstyle{seconddash}=[dashed,line cap=round, dash pattern=on 0.5pt off 1pt]

\newcommand{\drawWLD}[2]{

\pgfmathsetmacro{\n}{#1}
\pgfmathsetmacro{\radius}{#2}
\pgfmathsetmacro{\angle}{360/\n}
\draw (0,0) circle (\radius);
    \foreach \i in {1,2,...,\n} {
      \draw (\angle*\i:\radius) node {$\bullet$};
    }

}

\newcommand{\drawpolypart}[2]{
\pgfmathsetmacro{\n}{#1}
\pgfmathsetmacro{\radius}{#2}
\pgfmathsetmacro{\angle}{360/\n}
    \foreach \i in {1,2,...,\n} {
      \draw (\angle*\i+ \angle/2:\radius) node {$\bullet$};
     \pgfmathsetmacro{\x}{\angle*\i - \angle/2}
      \pgfmathsetmacro{\concave}{((\n-1.5)/\n)}
      \draw (\x:\radius cm) .. controls (\angle *\i: \concave* \radius cm) .. (\x + \angle:\radius cm);
    }

}

\newcommand{\drawprop}[4]{
\pgfmathsetmacro{\r}{#1}
\pgfmathsetmacro{\bumpr}{#2}
\pgfmathsetmacro{\s}{#3}
\pgfmathsetmacro{\bumps}{#4}
\pgfmathsetmacro{\perturbe}{\angle/\n}

\begin{scope}
\draw[propagator] (\angle*\r + \angle/2 + \bumpr*\perturbe:\radius) -- (\angle*\s + \angle/2 + \bumps*\perturbe:\radius);
\end{scope}
}

\newcommand{\drawlabeledprop}[5]{
\pgfmathsetmacro{\r}{#1}
\pgfmathsetmacro{\bumpr}{#2}
\pgfmathsetmacro{\s}{#3}
\pgfmathsetmacro{\bumps}{#4}
\pgfmathsetmacro{\perturbe}{\angle/\n}

\begin{scope}
\draw[propagator] (\angle*\r + \angle/2 + \bumpr*\perturbe:\radius) -- (\angle*\s + \angle/2 + \bumps*\perturbe:\radius) node[midway, below] {#5};
\end{scope}
}

\newcommand{\drawchord}[2]{
\pgfmathsetmacro{\r}{#1}
\pgfmathsetmacro{\s}{#2}

\begin{scope}
\draw (\angle*\r + \angle/2:\radius) -- (\angle*\s + \angle/2:\radius);
\end{scope}
}

\newcommand{\drawnumbers}{
  \foreach \i in {1,2,...,\n} {
  \pgfmathsetmacro{\x}{\angle*\i}
  \draw (\x:\radius*1.15) node {\footnotesize \i};
}
}

\newcommand{\drawnumbersshift}{
  \foreach \i in {1,2,...,\n} {
  \pgfmathsetmacro{\x}{\angle*\i + \angle/2}
  \draw (\x:\radius*1.15) node {\footnotesize \i};
}
}

\newcommand{\R}{\mathbb{R}}

\newcommand{\Gr}{\mathbb{G}_{\R, \geq 0}}

\newcommand{\rk}{\textrm{rk} }

\def\ba #1\ea{\begin{align} #1 \end{align}}
\def\bas #1\eas{\begin{align*} #1 \end{align*}}
\def\bml #1\eml{\begin{multline} #1 \end{multline}}
\def\bmls #1\emls{\begin{multline*} #1 \end{multline*}}
\newcommand{\cP}{\mathcal{P}}

\newcommand{\cB}{\mathcal{B}}

\newcommand{\Prop}{\textrm{Prop}}

\newtheorem{thm}{Theorem}[section]

\newtheorem{lem}[thm]{Lemma}
\newtheorem{cor}[thm]{Corollary}
\newtheorem{prop}[thm]{Proposition}

\theoremstyle{remark}
\newtheorem{eg}[thm]{Example}

\theoremstyle{definition}
\newtheorem{dfn}[thm]{Definition}
\newtheorem{rmk}[thm]{Remark}

\title{Combinatorics of the geometry of Wilson loop diagrams I: equivalence classes via matroids and polytopes}
\author{Susama Agarwala\thanks{SA was partially supported by an Office of Naval Research grant.}, Si\^an Fryer, and Karen Yeats\thanks{KY is supported by an NSERC Discovery grant, by the Canada Research Chair program, and also, over some of the time this work was developed, by a Humboldt Fellowship from the Alexander von Humboldt foundation.}}
\begin{document}
\maketitle

\begin{abstract}
  Wilson loop diagrams are an important tool in studying scattering amplitudes of SYM $N=4$ theory and are known by previous work to be associated to positroids. We characterize the conditions under which two Wilson loop diagrams give the same positroid, prove that an important subclass of subdiagrams (exact subdiagrams) corresponds to uniform matroids, and enumerate the number of different Wilson loop diagrams that correspond to each positroid cell.  We also give a correspondence between those positroids which can arise from Wilson loop diagrams and directions in associahedra.
\end{abstract}

\section{Introduction}

This paper is the first in a two part series investigating the combinatorics and geometry underlying SYM $N=4$ theory. The series lays out several results about the relationship between Wilson loop diagrams and the positroid cells that give a CW-complex structure to the positive Grassmannian $\Gr(k,n)$. Understanding the precise nature of this relationship is crucial in the study of scattering amplitudes in SYM $N=4$, as described below.

In \cite{wilsonloop} Agarwala and Marin-Amat showed that every weakly admissible Wilson loop diagram can be identified with a type of matroid called a positroid, and observed that this mapping is neither one-to-one nor onto. Each such positroid corresponds to a positroid cell in the appropriate positive Grassmannian $\Gr(k,n)$ \cite{Postnikov}. In this paper we give a simple characterization of this lack of injectivity in terms of an equivalence relation on the Wilson loop diagrams, and enumerate the number of different Wilson loop diagrams that correspond to each positroid cell. We also study the lack of surjectivity, and show that the number of distinct positroid cells associated to Wilson loop diagrams of a given size is given by counting the faces up to parallelism in an associahedron. 

The second paper \cite{generalcombinatoricsII} in this series identifies precisely which positroid cell is associated to a given Wilson loop diagram, by giving an algorithm for passing from a Wilson loop diagram to the Grassmann necklace of the associated positroid cell. In this manner we give a direct path from a Wilson loop diagram to a positroid cell, circumventing the painstaking process of examining sets of bases. With this identification in hand, we also prove that any Wilson loop diagram with $k$ propagators corresponds to a $3k$ dimensional positroid cell, and characterize the volume forms that Wilson loop diagrams associate to each positroid cell. The results of the first paper are not prerequisite for the second paper in this series.

In recent years, there has been an active program researching the geometry and combinatorics underlying SYM $N=4$ theory \cite{wilsonloop, Arkani-Hamed:2013jha, Amplituhedronsquared, galashinlam,AmplituhedronDecomposition}. This body of work started with the observation that the on shell (tree level) amplitudes of this theory correspond to the volume of a subspace of a positive Grassmannian, called an Amplituhedron \cite{Arkani-Hamed:2013jha}. Since then, there has been significant work studying the structure of the Amplituhedron both geometrically and combinatorially, for instance \cite{arkani:2012nw, Arkani-Hamed:2013kca, galashinlam,AmplituhedronDecomposition}. The focus of much of this work has been on understanding the relationship between BCFW diagrams and positroid cells, which are the combinatorial objects from the physics of the system and the corresponding natural objects in the positive Grassmannian respectively.

Meanwhile, there is a separate body of work studying SYM $N=4$ theory from the point of view of Feynman integrals in twistor space \cite{Adamo:2011pv, Britto:2005fq, Cachazo:2004kj}. These integrals are calculated via Wilson loops. As noted above, Agarwala and Marin-Amat uncovered a connection between the Feynman integrals developed in this literature and positroids, defined  as matroids that are realized as an element of the positive Grassmannian $\Gr(k,n)$ \cite{wilsonloop}. In separate papers with Fryer and with Marcott, this work was extended to study these Feynman integrals in terms of positroid cells and the positive Grassmannian \cite{casestudy, non-orientable}. Others have tried different approaches to define a geometric space associated to these integrals in a manner similar to the Amplituhedron \cite{Amplituhedronsquared, HeslopStewart}. 

Both the Amplituhedron literature and the Wilson loop literature associate $n$ particle $N^kMHV$ interactions in SYM $N=4$ theory to volume forms on a set of positroid cells of $\Gr(k,n)$. This gives a CW complex of some submanifold in $\Gr(k,n)$, and understanding the geometry and topology of these spaces is the subject of ongoing research. 

However, the two approaches are not identical: for example, in the case of Wilson loop diagrams with $k=2$ and $n=6$, past work of Agarwala and Fryer shows that the manifold is not contractible \cite{casestudy}. More generally, it has been shown that the spaces defined by the Amplituhedron and by Wilson loop diagrams are not the same \cite{Amplituhedronsquared}. In fact, it is in some ways more natural to consider Deodhar cells of $\mathbb{G}(k, n+1)$ for Wilson loop diagrams.  In this case, one gets that the manifold represented by the diagrams is not orientable \cite{non-orientable}. 

In this paper we focus on the correspondence between Wilson loop diagrams and positroids, which was first defined in \cite{wilsonloop}. Our main results are as follows. We first show that two Wilson loop diagrams define the same positroid if and only if they only differ by exact subdiagrams (Theorem~\ref{same matroid iff equiv}). Equivalently, this says that the map from Wilson loop diagrams to positroid cells is one-to-one exactly on the subset of diagrams with no non-trivial exact subdiagrams. We count the total number of Wilson loop diagrams associated to a positroid cell (Corollary~\ref{number of equiv diagrams}). Finally, we show that enumerating the number of distinct positroid cells associated to Wilson loop diagrams is equivalent to enumerating the number of directions of an associahedron (Theorem~\ref{thm:count inequiv diagrams}).

%In that paper, \hlfix{one of us}{see previous comment about self reference convention} with Amat show that there are certain families of diagrams, specifically those with exact subdiagrams, that do map to the same positroid cell. In \cite{casestudy} some of us explicitly show that not every positroid cell of a given dimension corresponds to a Wilson loop diagram. 

%In this paper, we characterize this correspondence more precisely. Specifically, we show that  We show that exact diagrams correspond to uniform matroids, and that contraction by exact subdiagrams has useful matroidal properties. Further, given a Wilson loop diagram with a non-trivial exact subdiagram, we count the total number of Wilson loop diagrams that correspond to that same positroid cell. While we do not count the number of distinct positroid cells that correspond to Wilson loop diagrams, we do show that this 

\subsection{Roadmap}

Section~\ref{section overall background} summarizes the required background, with definitions for Wilson loop diagrams given in Subsection~\ref{section WLD background} and matroids in Subsection~\ref{sec matroid background}. Key definitions include {\em admissible} Wilson loop diagrams (Definition~\ref{admisdfn}), {\em exact} subdiagrams (Definition~\ref{def:exact diagram}), and an equivalence relation on admissible Wilson loop diagrams (Definition~\ref{equivdfn}).  The connection between matroids and Wilson loop diagrams is stated precisely in Theorem \ref{thm WLD defines matroid}.

The main result of Section~\ref{sec equiv} is Theorem~\ref{same matroid iff equiv}, which says that two Wilson loop diagrams define the same matroid (positroid) if and only if they are equivalent.  In Subsection~\ref{sec: polygon partitions} we relate exact subdiagrams of Wilson loop diagrams to triangulations of polygons, and hence show that every Wilson loop diagram can be uniquely decomposed as a collection of exact subdiagrams (Corollary~\ref{uniqueproppartitioncor}). In Subsection~\ref{sec: exact diagram matroidal props} we examine the matroidal properties of the exact subdiagrams, and prove that a subdiagram of a Wilson loop diagram defines a uniform matroid if and only if the subdiagram is exact (Theorem \ref{exactuniformthm}). This uniform matroid is in fact equal to the dual matroid restricted to the exact subdiagram (Remark \ref{remark exact dual restiction}).  Finally, Subsection~\ref{sec: matroids and equivalence} is devoted to the proof of Theorem~\ref{same matroid iff equiv}, along with an easy corollary which gives a formula for the size of each equivalence class (Corollary~\ref{number of equiv diagrams}).

Having identified the equivalence classes and their sizes, the next natural question is to identify the number of inequivalent diagrams under this equivalence relation.  The main result of Section~\ref{sec associahedron} is Theorem \ref{thm:count inequiv diagrams}, which shows that the set of admissible Wilson loop diagrams (up to equivalence) is in bijection with the set of faces in the associahedron (up to parallelism).  Towards this result we go over some relevant background on polytopes in Subsection~\ref{sec polytope background}, and then prove the main result in Subsection~\ref{sec associahedron results} in the form of two propositions: Proposition~\ref{prop easy way} and Proposition~\ref{prop hard way}.

\section{Background}\label{section overall background}

This section gives the relevant physical and combinatorial background for this paper. Section \ref{section WLD background} defines Wilson loop diagrams as combinatorial objects, and presents relevant definitions and results from previous work. Section \ref{sec matroid background} presents some background on matroids and some existing results relating Wilson loop diagrams to matroids. 

\subsection{Wilson Loop diagrams}\label{section WLD background} 

%What are Wilson loop diagrams and their integrals.

\begin{dfn}\label{WLdfn}
A {\em Wilson loop diagram} is given by the following data: a cyclically ordered set $V$, along with a choice of first vertex (labeled $1$), and a set $\cP$ of $k$ pairs of elements of $V$ called {\em propagators}, written $\{p_r = (i_r, j_r)\}_{r=1}^k$. We write $W = (\cP,V)$ to denote the Wilson loop diagram. \end{dfn}

Wherever possible, we drop the subscripts on propagators and their identifying vertices. Propagators are undirected, so $p = (i,j) = (j,i)$. As in previous work, we follow the convention of writing $p = (i,j)$ with  $i +1 \leq j$ relative to the first vertex. 

A Wilson loop diagram can also be represented visually, which provides useful intuition for many of the results that follow. We depict the diagram $(\cP,V)$ as a circle with marked points, called {\em vertices}.  These vertices are labeled by $V$ (preserving the cyclic ordering). The arcs between consecutive vertices are called {\em edges}. There are also $k$ wavy lines in the interior of the diagram, each connecting a pair of edges.  These depict the propagators. Specifically, a propagator $p =(i,j)$ has one endpoint on the edge joining vertices $i$ and $i+1$ and another endpoint on the edge joining $j$ and $j+1$, where $i+1$ and $j+1$ denote the successor of $i$ and $j$ respectively in the cyclic order on $V$. 

To simplify language, we let the edges inherit a cyclic ordering from the vertices:

\begin{dfn}\label{def:edges}
The $i^{th}$ edge of $W$ is the arc that lies between the vertices $i$ and $i+1$.
\end{dfn}

Thus we may speak of a propagator $p = (i, j)$ as being supported by the $i$th and $j$th edges.

Note that the marked circle captures the cyclic ordering on $V$, and the choice of a first vertex gives it a compatible linear order. Both the cyclic and the linear order become the correct perspective at various points in this paper. 

Often we take $V$ to be the cyclically ordered set of integers $\{1, \ldots, n\}$, which we denote by $[n]$. In this case we write $W = (\cP, [n])$, and all indices are considered modulo $n$.

\begin{eg}\label{eg of WLD} Let $V = [8]$ and $\cP = \{ (1,4), (2, 4), (5, 8)\}$. Then $W = (\cP,[8])$ is the Wilson loop diagram
\bas \begin{tikzpicture}[rotate=67.5,baseline=(current bounding box.east)]
	\begin{scope}
	\drawWLD{8}{1.5}
	\drawnumbers
	\drawprop{1}{0}{4}{0}
	\drawprop{2}{0}{4}{-1}
\drawprop{5}{0}{8}{0}
		\end{scope}
	\end{tikzpicture} \; . \eas 
\end{eg}

\medskip

We introduce some notation to speak of vertices supporting a propagator, and the set of propagators supported on a vertex set.

\begin{dfn} \label{VPropdfn}
Let $W = (\cP, [n])$. In this definition and in the sequel, for any $P \subseteq \cP$ or $U \subseteq [n]$, we denote $P^c = \cP \setminus P$ and $U^c = [n] \setminus U$.
\begin{enumerate}
\item For $p = (i,j) \in \cP$, define $V(p) = \{i, i+1, j, j+1\}$ to be the set of vertices supporting $p$. Similarly, for $P \subseteq \cP$ we define the set $V(P) = \cup_{p \in P} V(p)$ to be the vertex support of $P$.
\item For $U \subseteq [n]$, write $\Prop(U) = \{ p \in \cP  \ | \ V(p) \cap U \neq \emptyset \} $ to denote the set of propagators which are at least partially supported on $U$.
\item For $P \subseteq \cP$, define the {\em propagator flat}\footnote{This choice of terminology is justified in Lemma~\ref{lem facts about WLD matroids} below.} $F(P) = V(P^c)^c$ to be the set of vertices in $[n]$ that do not support any propagators outside of the set $P$. Here, as throughout, we are using the superscript $c$ for the set theoretic complement.
\item The set of vertices that do not support any propagators is denoted $F(\emptyset)$. Vertices in this set are called {\em non-supporting}.  
\end{enumerate}
\end{dfn}

\begin{rmk}\label{alt F(P) rmk}
An equivalent definition of $F(P)$ is \bas F(P)  = \big(V(P) \setminus V(P^c) \big) \cup \{\text{all non-supporting vertices in $W$}\} \;,\eas 
i.e. $F(P)$ can be thought of as the set of vertices that either only support propagators in $P$ or do not support any propagators at all. Furthermore, note that by construction we have $\Prop(F(P)) \subseteq P$.  This containment is strict if there is any propagator $p$ in $P$ with the property that every vertex in the support of $p$ also supports a propagator not in $P$. 
\end{rmk}

We now define the notion of an admissible diagram, which imposes certain density and non-crossing conditions on the configuration of the propagators. The motivation for these conditions comes from the link to SYM $N=4$ theory, where the admissible Wilson loop diagrams are exactly those that correspond to $N^kMHV$ diagrams; see for example \cite[section 2]{HeslopStewart}.

\begin{dfn}\label{admisdfn}
A Wilson loop diagram $W = (\cP,V)$ is {\em admissible} if it satisfies the following three conditions: \begin{enumerate}
\item $|V| \geq |\cP| + 4,$
\item for any non-empty set of propagators, $Q \subseteq \cP$, $Q \neq \emptyset$, one has $|V(Q)| \geq |Q| + 3$.\,
\item there does not exist a pair of propagators $(i_p,j_p),(i_q,j_q) \in \cP$ such that $i_p < i_q < j_p <j_q$.
\end{enumerate}
A Wilson loop diagram is called {\em weakly admissible} if the second and third conditions hold.
 \end{dfn}

The first condition restricts the overall density of propagators in the diagram, while the second imposes an upper bound on how densely the propagators can be fitted into any portion of the diagram. The third ensures that ensures that no propagators cross in the interior of the diagram. 

There are a few things to note here.  First, what we call weakly admissible here is called admissible in (\cite{wilsonloop}, see definition 1.11 and section 1.3).
Note also that if two propagators have the same pair of supporting edges, or if a propagator is supported on two adjacent edges, then the Wilson loop diagram is not admissible or even weakly admissible.
Further, if we take any admissible Wilson loop diagram and remove the non-supporting vertices, we obtain a weakly admissible Wilson loop diagram that may or may not be admissible itself.

For the remainder of the paper, we will restrict our attention to admissible and weakly admissible Wilson loop diagrams.

\begin{dfn} \label{subdiagramdfn}
Let $W = (\cP, [n])$ be an admissible Wilson loop diagram. The weakly admissible diagram $W' = (P,V)$ is a {\em subdiagram} of $W$, written $W' \subseteq W$, if \bas P \subseteq \cP \text{ and } \quad V(P) \subseteq V \subseteq [n]\;.\eas
\end{dfn}

There is one particular type of subdiagram that deserves special attention:

\begin{dfn}\label{def:exact diagram}
For $W$ a weakly admissible diagram, the subdiagram $(P, V(P))$ is {\em exact} if $|V(P)| = |P| + 3$.
\end{dfn}

Note that by Definition \ref{admisdfn}, any single propagator and its support defines an exact subdiagram. These are called \emph{trivial} exact subdiagrams.

We will see in Section~\ref{sec equiv} that the exact diagrams play an important role in classifying admissible Wilson loop diagrams.

\vspace{3pt}
\begin{eg}
The Wilson loop diagram $W$ in Example~\ref{eg of WLD} is an example of an admissible diagram. It has a single non-supporting vertex, namely the vertex labeled 7, so we have $F(\emptyset) = \{7\}$. For the propagator $p = (5, 8)$, we have $V(p) = \{5,6,8,1\}$ and $F(p) = \{6, 7,8\}$: the vertices 1 and 5 are excluded from $F(p)$ since they support other propagators as well, and 7 is included (despite not being in the support of $p$) since it supports no propagators at all. Finally, the subdiagram $(\{(2, 4), (1, 4)\}, \{1, 2, 3, 4, 5\})$ is an example of an exact subdiagram.  In fact it is the only non-trivial exact subdiagram in $W$.
\end{eg}

The exact subdiagrams define an equivalence relation amongst Wilson loop diagrams as follows.

\begin{dfn}\label{equivdfn} 
Let $W = (\cP,[n])$ and $W' = (\cP',[n])$ be two weakly admissible subdiagrams. Let $\sim$ denote the transitive closure of the following relation: $W$ and $W'$ are related if
\begin{enumerate}
\item there exist exact subdiagrams $(P, V(P)) \subseteq W$ and $(P', V(P')) \subseteq W'$ with the same number of propagators and supported on the same vertices, i.e. $|P| = |P'|$ and  $V(P) =  V(P')$.
\item the complementary subdiagrams (i.e. all other propagators and vertices) are identical, i.e. $(\cP \setminus P, [n]) = (\cP' \setminus P', [n])$.
\end{enumerate}
The relation $\sim$ is an equivalence relation on the set of all weakly admissible Wilson loop diagrams with vertex set $[n]$.
\end{dfn}

\begin{eg} \label{eg:equivdiags}
Note that we genuinely do need to take the transitive closure of this relation in order to obtain an equivalence relation. For example, consider the following three Wilson loop diagrams: 
\bas
W_1 = \begin{tikzpicture}[rotate=67.5,baseline=(current bounding box.east)]
	\begin{scope}
	\drawWLD{8}{1.5}
	\drawnumbers
	\drawlabeledprop{1}{0}{4}{0}{\footnotesize $q$ \quad}
	\drawlabeledprop{2}{0}{4}{-1}{\footnotesize  $p$}
    \drawlabeledprop{5}{0}{8}{0}{\footnotesize $r$ \ \ }
    \drawlabeledprop{5}{1}{7}{0}{\footnotesize \quad \ $s$}
		\end{scope}
	\end{tikzpicture} \quad
W_2 = \begin{tikzpicture}[rotate=67.5,baseline=(current bounding box.east)]
	\begin{scope}
	\drawWLD{8}{1.5}
	\drawnumbers
	\drawlabeledprop{1}{0}{4}{0}{\footnotesize $q$ \quad}
	\drawlabeledprop{2}{0}{4}{-1}{\footnotesize $p$}
    \drawlabeledprop{5}{0}{8}{0}{\footnotesize $u$ \ \ }
    \drawlabeledprop{6}{0}{8}{-1}{\footnotesize $v$ \ }
		\end{scope}
	\end{tikzpicture} \quad
W_3 = \begin{tikzpicture}[rotate=67.5,baseline=(current bounding box.east)]
	\begin{scope}
	\drawWLD{8}{1.5}
	\drawnumbers
	\drawlabeledprop{1}{0}{4}{0}{\footnotesize $a$ \quad}
	\drawlabeledprop{1}{1}{3}{0}{\footnotesize  $b$ \ \ \ }
    \drawlabeledprop{5}{0}{8}{0}{\footnotesize $u$ \ \ }
    \drawlabeledprop{6}{0}{8}{-1}{\footnotesize $v$ \quad }
		\end{scope}
	\end{tikzpicture} .
\eas

$W_1$ and $W_2$ are related because they are identical except on the exact subdiagrams 
\[(\{r, s\}, \{5,6,7,8,1\}) \text{ and }(\{u, v\}, \{5,6,7,8,1\})\] respectively. Similarly, $W_2$ and $W_3$ are related due to the exact subdiagrams 
\[(\{p, q\}, \{1,2,3,4,5\}) \text{ and }(\{a, b\}, \{1,2,3,4,5\}).\] 
However, the difference between $W_1$ and $W_3$ is the union of two exact subdiagrams, which is not itself exact.

\end{eg}
\medskip

As described in the introduction, Wilson loop diagrams come from the study of SYM $N = 4$ in the physics literature. In particular (see \cite[section 2]{HeslopStewart} for details) this gives that each diagram $W$ is associated to a matrix of data $C(W)$, which we now describe. 

\begin{dfn}\label{def:CWmatrix}
Each Wilson loop diagram $W = (\cP, [n])$ with $|\cP| = k$ is associated to a $k \times n$ matrix $C(W)$ with non-zero real variable entries, defined by:
\bas C(W)_{p,q} = \begin{cases} c_{p,q} & \textrm{ if } q \in V(p) \\
0  & \textrm{ if } q \not \in V(p)  \end{cases}
\;. \eas
\end{dfn}

\begin{eg}
If we order the propagators of $W_1$ from Example \ref{eg:equivdiags} as follows: \bas (1,4), \, (2,4), \, (5,7), \, (5,8) \eas then we obtain the corresponding matrix
\bas C(W_1) = \left(
\begin{array}{cccccccc}
c_{1,1} & c_{1,2} & 0 & c_{1,4} & c_{1,5} & 0 & 0 & 0 \\
0 & c_{2,2} & c_{2,3} & c_{2,4} & c_{2,5} & 0 & 0 & 0 \\
0 & 0 & 0 & 0 & c_{3,5} & c_{3,6} & c_{3,7} & c_{3,8} \\
c_{4,1} & 0 & 0 & 0 & c_{4,5} & c_{4,6} & 0 & c_{4,8}  \\
\end{array}
\right) \;.\eas

\end{eg}

%The Wilson loop diagrams also define a volume form on $\Sigma(W)$: \bas \Omega(W) = \frac{\prod_{r=1}^{|\cP|} \prod_{v \in V_{p_r}} \textrm{d}c_{p_r}}{R(W)} \;. \eas The denominator $R(W)$ is a polynomial defined by $2 \times 2$ and $1 \times 1$ minors of $C(W)$ as defined below.
%
%\begin{dfn}\label{def R(W)}
%For $W = (\cP, [n])$, $R(W) = \prod_{e=1}^n R_e$, with $R_e$ defined by the propagators ending on it. For any edge $e$ of $W$, order the propagators incident on $e$ as $\{p_1 \ldots p_r\}$, ordered such that $p_1$ is closest to the vertex $e$, $p_r$ closest to $e+1$, and $p_i$ is closer to $e$ than $p_{i+1}$. Then \bas R_e =  c_{p_1,e+1} \prod_{j= 1}^{r-1} \left((c_{p_j,e} c_{p_{j+1},e+1} - c_{p_{j+1},e} c_{p_{j},e+1} ) \right) c_{p_s,e}\;.\eas Note that in this notation, if $r = 1$, $R_e = c_{p,e} c_{p,e+1}$.
%\end{dfn}

\subsection{Wilson loop diagrams as matroids\label{sec matroid background}}

The matrix $C(W)$ allows us to associate a matroid to each Wilson loop diagram. We first give a quick summary of the matroid terminology that we will need.  It is not intended as a comprehensive introduction to matroids and the interested reader is referred to \cite{OxleyMatroidBook}.

A {\em matroid} $M = (E,\cB)$ consists of a finite ground set $E$ and a non-empty family $\cB$ of subsets of $E$, such that elements of $\cB$ satisfy the {\em basis exchange property}: for any distinct $B_1,B_2 \in \cB$ and any $a \in B_1 \setminus B_2$, there exists some $b \in B_2 \setminus B_1$ such that $(B_1 \setminus \{a\})\cup \{b\} \in \cB$. The elements of $\cB$ are called the {\em bases} of the matroid. Note that the basis exchange property immediately implies that all bases have the same size.

A subset $A \subseteq E$ is called {\em independent} in $M$ if $A \subseteq B$ for some $B \in \cB$, and is called {\em dependent} otherwise. (In other words, a basis is simply an independent set of maximal size in $M$.) The {\em rank}  $\rk(A)$ of a subset $A \subseteq E$ is the size of the largest independent set contained in $A$. The rank of the matroid itself is defined to be $\rk(E)$.

A {\em circuit} in $M$ is a minimally dependent set. That is, it is a set $C \subseteq E$ such that $C$ is dependent but $C \setminus \{e\}$ is independent for any $e \in C$. A union of circuits is called a {\em cycle}. On the other hand, a {\em flat} is a maximally dependent set, i.e. a set $F \subseteq E$ such that $\rk(F \cup \{e\}) = \rk(F) + 1$ for any $e \in E \setminus F$. Unsurprisingly, a {\em cyclic flat} is a set which is both a flat and a cycle. The set of circuits in a matroid uniquely defines that matroid, as does the set of flats.  Thus one could specify a matroid by listing either its independent sets, its bases, its circuits, or its flats.

Finally, we describe several important types of matroids. A matroid of rank $k$ with a ground set of size $n$ is called {\em realizable} if there exists some matrix $A$ in the $k\times n$ real Grassmannian $\mathbb{G}_{\mathbb{R}}(k,n)$ whose non-zero $k\times k$ minors are exactly those with columns indexed by an element of $\cB$. If all the maximal minors of the matrix are nonnegative then the matroid realized by this matrix is called a {\em positroid}.  In other words a positroid is a matroid\footnote{Note that matroid isomorphism allows arbitrary permutations of the ground set, while positroid isomorphism only allows cyclic permutations in order to preserve the nonnegativity of the minors.  In particular, the property of being a positroid is not preserved under matroid isomorphism.} that is realized by an element of the totally nonnegative Grassmannian $\Gr(k,n)$. Finally, a {\em uniform matroid} of rank $r$ is a matroid in which any set of size $\leq r$ is independent.

Matroid theory relates to the study of Wilson loop diagrams as follows. In \cite[section 3]{wilsonloop}, Agarwala and Marin-Amat show that every weakly admissible Wilson loop diagram $W$ with $k$ propagators defines a positroid $M(W)$ of rank $k$ on a base set of size $n$. It is a direct consequence of Postnikov's work, {\cite[Theorem 6.5]{Postnikov}}, that positroid cells are convex subset of $\Gr(k,n)$. Putting these together gives that the Wilson loop diagrams  correspond to convex subsets of $\Gr(k,n)$, connecting the geometry of Wilson loop diagrams with the geometry underlying the Amplituhedron. 

The independent sets of $M(W)$ can be read directly from the diagram:

\begin{thm} \label{thm WLD defines matroid} \cite[Theorem 3.6]{wilsonloop} Any Wilson loop diagram $W =(\cP, [n])$ defines a matroid $M(W)$ with ground set $[n]$. The independent sets are exactly those subsets $V \subseteq [n]$ such that every $U\subseteq V$ satisfies $|\Prop(U)| \geq |U|$. \label{thm:WLDmatroid}\end{thm}
In other words, the independent sets of $M(W)$ correspond to the sets of vertices in $W$ with ``enough propagators sufficiently well distributed'' on them: not only should $V$ itself support at least as many propagators as the number of vertices it contains, but the same should be true of any subset of $V$ as well.

Throughout, we take the {\em matroid defined by $W$} to be the matroid $M(W)$ of Theorem \ref{thm WLD defines matroid}. Note from Definition~\ref{def:CWmatrix} that the vertices of the diagram $W$ correspond to columns of the associated matrix $C(W)$. Therefore, for any subset of vertices $V \subseteq [n]$ we can consider the restriction $C(W)|_V$ of $C(W)$, which consists of the columns corresponding to $V$. Since the non-zero entries of $C(W)$ are independent real variables, if $C(W)|_V$ has fewer than $|V|$ rows of non-zero entries then it does not have full rank. This is equivalent to saying that if $V$ is a set of columns of $C(W)$ with $|\Prop(V)| < |V|$, then $V$ is a dependent set. Any set of columns containing such a $V$ cannot be independent. Comparing this to Theorem \ref{thm:WLDmatroid}, we see that $M(W)$ is also the matroid realized by $C(W)$. As mentioned above, Theorem [3.38] of \cite{wilsonloop} shows that if $W$ is weakly admissible, the matroid $M(W)$ is a positroid. Furthermore, Theorem 5.1 \cite{WLDdim} shows that the space parametrized by the matrices $C(W)$ matches the positroid cells defined by $M(W)$, up to a set of measure $0$. Note, this does not imply that every possible specializaton of $C(W)$ to a real valued matrix leads to a postroid, only that the set of values that do define a positroid is volume filling in $\R^{4k}$.  

% The identification of Theorem \ref{thm WLD defines matroid} is not 1-1: by \cite[Theorem 1.18]{wilsonloop}, if two admissible Wilson loop diagrams are equivalent (as in Definition~\ref{equivdfn}) then they define the same matroid. 

Let $W = (\cP,[n])$ be a weakly admissible Wilson loop diagram, and $M(W)$ its associated matroid. Where it will not cause confusion we conflate the two objects, we identify vertices of $W$ with elements of the ground set $[n]$ in $M(W)$. 

In particular, this allows us to prove results about $M(W)$ by considering the behavior of propagators in $W$. We record a few facts about the ranks of cycles and flats in $M(W)$ as an example of this.

\begin{lem}\label{lem facts about WLD matroids}
Let $W = (\cP,[n])$ be a weakly admissible Wilson loop diagram. The following are true.
\begin{enumerate}
\item The rank of a set $V \subseteq [n]$ is bounded above by $\min\{|V|,|\Prop(V)|\}$, with $\rk(V) = |V|$ if and only if $V$ is an independent set.
%\item Let $v \in [n]$ and $q \in \Prop(v)$. Then for any $V \subset [n]$ such that $q \not \in \Prop(V)$, we have $\rk(V \cup \{v\}) = \rk(V) +1$, i.e. adding a vertex that supports a new propagator to a set increases the rank of the set.
\item If $C \subseteq [n]$ is a cycle in $M(W)$, then $\rk(C) = |Prop(C)|$.
\item If $[n]$ can be partitioned into at least two non-empty sets, each of which support disjoint sets of propagators and in such a way that we obtain a partition of the propagator set, i.e. if we have \bas [n] = \sqcup_i V(P_i) \ \textrm{ such that } \ \sqcup P_i = \cP; \   \text{ with }V(P_i) \cap V(P_j) = \emptyset \text{ and } P_i \cap P_j = \emptyset \quad \forall i \neq j\;, \eas then the matroid $M(W)$ is separable, i.e. \bas M(W) = \bigoplus_i M(P_i, V(P_i)) \;.\eas 
\item The set $F(P)$ from Definition~\ref{VPropdfn} is a flat of $M( W)$, which we call the \emph{propagator flat} of $P$.
\item If $F$ is a cyclic flat of $M(W)$, then it is a propagator flat.
\end{enumerate}
\end{lem}
\begin{proof}
The first part of item 1 is \cite[Equation (9)]{wilsonloop} and surrounding discussion, and the second part is standard matroid theory. Item 2 is \cite[Lemma 3.27]{wilsonloop}. Item 3 is a direct consequence of \cite[Theorem 3.20]{wilsonloop} and the fact that $F(P_1)^c = V(P_1^c)$, and item 5 is given in \cite[Lemma 3.28]{wilsonloop}.

To prove item 4, we need to show that $F(P)$ is maximally dependent. If $F(P) = [n]$ then this is automatic, so suppose not and let $v \in [n] \setminus F(P)$. In other words, $v \in V(P^c)$ and so $v$ supports some propagator $q \not\in P$.  Let $S \subseteq F(P)$ be an independent set of maximal size. Then $\Prop(S) \subseteq P$ (from the definition of $F(P)$) and no subset of $S$ supports fewer propagators than the number of vertices it contains (by Theorem~\ref{thm WLD defines matroid}). Since $v$ supports a new propagator $q \not\in P$, the set $S \cup \{v\} \subseteq F(P) \cup \{v\}$ also satisfies this independence condition.

Since $S$ is independent, it follows from the first point in this lemma that $|S| = \rk(S) = \rk F(P)$. Similarly, 
\[|S \cup \{v\}| = |S|+1 = \rk (S \cup \{v\}) \leq \rk(F(P) \cup \{v\}) \leq \rk F(P) + 1.\]
Thus $\rk(F(P) \cup\{v\}) = \rk F(P) + 1$, as required.
\end{proof}

\begin{rmk}
By item 4 of Lemma~\ref{lem facts about WLD matroids}, the set $F(\emptyset)$ of non-supporting vertices in $W$ is a flat. Indeed, it is the unique flat of rank 0 in $M(W)$. Every element of $F(\emptyset)$ defines a circuit of size 1, and all other circuits of $M(W)$ must contain at least two elements.
\end{rmk}

\section{Matroids, triangulations, and equivalence classes of Wilson loop diagrams\label{sec equiv}} 

This section explores the correspondence between Wilson loop diagrams and positroids that was established in \cite[Theorem 3.38]{wilsonloop}. In that paper, the authors show that admissible Wilson loop diagrams define matroids that are also positroids. Furthermore, they show that if two Wilson loop diagrams are equivalent, then they define the same matroid \cite[Theorems 1.18 and 3.41]{wilsonloop}.

In this section we develop this analysis further, showing the converse of this result from \cite{wilsonloop} (see Theorem \ref{same matroid iff equiv}). We also enumerate the number of diagrams that map to a particular matroid in Corollary \ref{number of equiv diagrams}. To do this, we establish a connection between Wilson loop diagrams and polygon dissections in Subsection \ref{sec: polygon partitions}, and uncover some matroidal properties of exact subdiagrams in Subsection \ref{sec: exact diagram matroidal props}. This includes the fact that exact subdiagrams correspond to uniform matroids (Theorem \ref{exactuniformthm}), and that the associated matroid of an exact subdiagram can be written as a contraction by the complementary propagator flat (Theorem \ref{exact diagrams contractions}).

\subsection{Polygon dissections of Wilson loop diagrams\label{sec: polygon partitions}}

The equivalence relation on Wilson loop diagrams (Definition \ref{equivdfn}) is defined in terms of exact subdiagrams. Thus in order to understand the equivalence, we need a way to extract and compare exact subdiagrams. We do this via the notion of a polygon dissection of $W$.

\begin{dfn}\label{WLDtriangulationdfn}
  Let $W = (\cP, [n])$ be a weakly admissible Wilson loop diagram.  The \emph{polygon dissection} associated to $W$, denoted $\tau(W)$, is defined as follows.
  \begin{itemize}
  \item The vertices of $\tau(W)$ correspond to the edges of $W$.
  \item Labeling the vertices of $\tau(W)$ with the edge number of $W$ gives a cyclic order to the vertices. Connecting consecutive vertices gives a graph theoretic cycle called the polygon of $\tau(W)$.
  \item Each propagator of $W$ defines a diagonal edge of $\tau(W)$.  Specifically,  a propagator $(i,j) \in \cP$ defines a diagonal connecting the vertices $i$ and $j$ in $\tau(W)$.
  \end{itemize}
\end{dfn}

\emph{Polygon dissection} is a common term in combinatorics for a decomposition of a convex polygon into smaller polygons by means of noncrossing diagonals; see for example the introduction of {\cite{PSpolygon}}.  Each $\tau(W)$ is such an object.

\begin{lem}\label{tausimpleplanarlem}
If $W = (\cP, [n])$ is a weakly admissible Wilson loop diagram, then $\tau(W)$ is a simple planar graph whose outer face is a cycle. It inherits an embedding from $W$ such that the vertices all lie on this infinite face\footnote{That is, it is an \emph{outerplanar} graph}. These vertices are cyclically ordered, with a choice of first vertex giving it an additional compatible linear order.
\end{lem}

\begin{proof}
Since the vertices of $\tau(W)$ are labeled by the edges of $W$, which are cyclically ordered, this gives an ordering to the vertices and the outer face of $\tau(W)$ is a cycle.  Since $W$ does not admit any propagators of the form $p = (i, i+1)$ due to its admissibility, there is exactly one edge connecting any two adjacent edges of $\tau(W)$. Similarly, there does not exist two propagators $p,q$ such that both $p$ and $q$ start at edge $i$ and end at edge $j$. Therefore, no two vertices of $\tau(W)$ can be connected by more than one edge, and so $\tau(W)$ is a simple graph. Finally, the embedding of $\tau(W)$ is induced from the embedding of the graph $W$, and since $W$ is weakly admissible we know that no pair of propagators can cross. Therefore, it is a planar embedding. 
 \end{proof}

\begin{rmk}
The previous lemma shows that $\tau$ does map into polygon dissections.  In fact $\tau$ is a bijection from weakly admissible Wilson loop diagrams to polygon dissections.  To see the inverse map, and hence obtain bijectivity, we need only reverse the construction above.  The only part which is nontrivial is that if diagonals of the dissection meet at a common vertex, we need to order the propagators associated to these diagonals along the edge associated to the vertex of the dissection, but using the noncrossing property of propagators this can be done uniquely.

The $\tau$ bijection will be useful, as is typical in bijective combinatorics, because some features are more visible on the dissection side.  In particular the exact subdiagrams are visible as triangulated pieces and equivalence comes down to retriangulation.  All of these notions are defined and explained below.

There are many other instances of this basic methodology in combinatorics, but we'd like to point out one in particular that involves many objects closely related to those we use here, \cite{CGdyck}\footnote{Thanks to the referee for the companion paper, \cite{generalcombinatoricsII}, for pointing this out.}.  In particular they, like us, look at a special case within the world of positroids and related objects which is enumerated by Catalan numbers.  For them these are positroids coming from unit interval graphs, while for us they are the equivalent exact subdiagrams.  Our objects are different, which can be seen by the fact that the exact subdiagram with 5 vertices (which is unique up to rotation) does not give a decorated permuation in the form of the Main Theorem of \cite{CGdyck}, but there are similarities in flavour with our arguments. 

As a final note before we begin laying out the results we need on $\tau$, if a Wilson loop diagram $W$ is weakly admissible but not admissible, then the same equality that would give exactness in a subdiagram holds, and so for the same reasons as described below for exact subdiagrams $\tau(W)$ will have all triangle internal faces and this is the only way that $\tau(W)$ can have all triangle internal faces.  Therefore, $\tau$ also gives a bijection between admissible Wilson loop diagrams and polygon dissections with at least one non-triangle internal face.
\end{rmk}

\begin{eg}\label{WLDtopolygonpartition}
In this example we return to two of the Wilson loop diagrams in Example \ref{eg:equivdiags}. We can pair diagrams with their polygon dissections as follows:
\bas W_1 = \begin{tikzpicture}[rotate=67.5,baseline=(current bounding box.east)]
	\begin{scope}
	\drawWLD{8}{1.5}
	\drawnumbers
	\drawprop{1}{0}{4}{0}
	\drawprop{2}{0}{4}{-1}
    \drawprop{5}{0}{8}{0}
    \drawprop{5}{1}{7}{0}
		\end{scope}
	\end{tikzpicture} \quad; \quad
\tau(W_1) = \begin{tikzpicture}[rotate=67.5,baseline=(current bounding box.east)]
	\begin{scope}
	\drawpolypart{8}{1.5}
    \drawnumbersshift
    \drawchord{1}{4}
    \drawchord{2}{4}
    \drawchord{5}{8}
    \drawchord{5}{7}
	\end{scope}
	\end{tikzpicture}
\eas and
\bas W_3 = \begin{tikzpicture}[rotate=67.5,baseline=(current bounding box.east)]
	\begin{scope}
	\drawWLD{8}{1.5}
	\drawnumbers
    \drawprop{1}{0}{4}{0}
	\drawprop{1}{1}{3}{0}
    \drawprop{5}{0}{8}{0}
    \drawprop{6}{0}{8}{-1}	
		\end{scope}
	\end{tikzpicture}\quad; \quad
\tau(W_3) = \begin{tikzpicture}[rotate=67.5,baseline=(current bounding box.east)]
	\begin{scope}
	\drawpolypart{8}{1.5}
    \drawnumbersshift
    \drawchord{1}{4}
    \drawchord{1}{3}
    \drawchord{5}{8}
    \drawchord{6}{8}
	\end{scope}
	\end{tikzpicture} .
\eas

\end{eg}

A planar embedding of a graph is a \emph{triangulation} if all faces, except possibly the infinite face, are triangles.

\begin{dfn}
  Let $W$ be a weakly admissible Wilson loop diagram and $\tau(W)$ its polygon dissection. A \emph{triangulated piece} of $\tau(W)$ is a connected, bridgeless subgraph of $\tau(W)$ which is a triangulation. We will take the convention that a subgraph consisting of a single diagonal edge is called a \emph{trivial} triangulated piece.
A {\em maximal} triangulated piece is one which is not contained in any strictly larger triangulated piece.
\end{dfn}

\begin{eg} \label{eg: unique decomposition} For the Wilson loop diagrams and polygon dissections in Example \ref{WLDtopolygonpartition}, the vertex sets $\{1, 2, 3,4\}$ and $\{5, 6, 7, 8\}$ give maximal triangulated pieces for both $\tau(W_1)$ and $\tau(W_3)$. The vertex set $\{4,5, 8, 1\}$ is not a triangulation in either polygon dissection. 
\end{eg}

\begin{dfn}
 A {\em decomposition} of a polygon dissection $\tau(W)$ is a set of bridgeless, connected, induced subgraphs of $\tau(W)$ which partition the edges of $\tau(W)$.  
\end{dfn}

\begin{lem} \label{decompositionlem}
  For $W$ a weakly admissible Wilson loop diagram, the polygon dissection $\tau(W)$ has a unique decomposition into maximal triangulated pieces and edges in the polygon of $\tau(W)$.  We call this a \emph{maximal decomposition}.
\end{lem}

\begin{proof}
We begin by giving an algorithm for the decomposition, then prove its uniqueness. Let $W = (\cP, [n])$, with $|\cP| = k$.

By \emph{splitting} a vertex $v$ we will mean replacing $v$ by new vertices $v_1, v_2,\ldots, v_{\text{deg}(v)}$ such that each $v_i$ has exactly one neighbour and the union of the $v_i$ is the neighbourhood\footnote{The neighbourhood of a vertex is the set of adjacent vertices} of $v$.

Let $T(W)$ be the dual graph of $\tau(W)$ with the vertex corresponding to the infinite face split.
%In other words, place a vertex on each finite face of $\tau(W)$. These vertices are connected if there is a chord edge separating the edges. Furthermore, if the face is bounded by an edge of the polygonal cycle of $\tau(W)$, the corresponding vertex gets a leaf edge for each such boundary.
Since $\tau(W)$ is an embedded graph (with a fixed distinguished embedding) by Lemma \ref{tausimpleplanarlem}, $T(W)$ is an uniquely defined graph.

Furthermore $T(W)$ is a tree because it is connected, has $n+k+1$ vertices ($k+1$ from the internal faces of $\tau(W)$ and $n$ from the outer face) and $n+k$ edges (since $\tau(W)$ has $n+k$ edges).  Additionally, since $\tau(W)$ is simple, $T(W)$ has no vertices of degree $2$.  

%We claim that $T(W)$ is a tree. Since $\tau(W)$ is a planar graph with $k$ chord edges and no internal vertices, there are $k+1$ internal faces of $\tau(W)$. Each of these internal faces corresponds to a vertex of $T(W)$, and there are $n$ vertices of $T(W)$ from the splitting of the dual graph, so $T(W)$ has $n+k+1$ vertices in total. On the other hand, $T(W)$ has $n+k$ edges, corresponding to the $n +k$ edges of $\tau(W)$. Therefore, $T(W)$ is a tree.

Now split every vertex of $T(W)$ which has degree $>3$; see Figure~\ref{fig:vertex splitting in the dual graph} for an example of this. That is, we split $T(W)$ at every vertex corresponding to a non-triangle face of $\tau(W)$. 
%By construction, each face of $\tau(W)$ is a polygon. Therefore, each vertex of $T(W)$ is at least 3 valent (one for each edge of the polygon). \hlfix{If split each vertex of $T(W)$ if it is strictly greater than $3$ valent.}{???} (In other words, we fail to split exactly when the face is a triangle).

We claim that the connected components of $T(W)$ correspond to the decomposition of $\tau(W)$ into maximal triangulated pieces and edges originally in the polygon of $\tau(W)$.  To see this, let $f$ be the forest thus obtained.
The vertices of $f$ either have degree 1 or 3.
%, (if they correspond leaf vertices) or 3,  (if they correspond to triangular faces).
Trees of $f$ with no trivalent vertices correspond to either edges in the polygon of $\tau(W)$ (if they were originally leaves of $T(W)$) or to maximal trivial triangulated pieces.
Splitting at all the faces that are not triangles ensures maximality of the decomposition. If the splitting were not maximal, then one could add a triangle to a connected component of the splitting, but this would imply that that splitting happened at a degree $3$ vertex.

\begin{figure}[h]
\bas 
\begin{tikzpicture}[rotate=67.5,baseline=(current bounding box.east)]
	\begin{scope}
	\drawpolypart{8}{1.5}
    \drawnumbersshift
    \drawchord{1}{4}
    \drawchord{1}{3}
    \drawchord{5}{8}
    \drawchord{6}{8}
 \draw[gray] (0,0) node {$\bullet$};
\draw[gray] (-1,.2) node {$\bullet$};	
\draw[gray] (-.5,1.1) node {$\bullet$};
\draw[gray] (.8,-.3) node {$\bullet$};	
\draw[gray] (1.1,-.5) node {$\bullet$};
\draw[gray](1.1,-.5) -- (.8,-.3) --(0,0)--(-1,.2) -- (-.5,1.1);
\foreach \i in {1,2,...,8} {
      \draw[gray] (45*\i:2) node {$\bullet$};
    }
\draw[gray] (45*1:2) --(0,0) -- (45*5:2);
\draw[gray] (45*6:2) -- (.8,-.3);
\draw[gray] (45*8:2) -- (1.1,-.5) -- (45*7:2);
\draw[gray] (45*2:2) -- (-.5,1.1) -- (45*3:2);
\draw[gray] (45*4:2) -- (-1,.2);
	\end{scope}
	\end{tikzpicture} \qquad \longrightarrow \qquad 
\begin{tikzpicture}[rotate=67.5,baseline=(current bounding box.east)]
	\begin{scope}
	\drawpolypart{8}{1.5}
    \drawnumbersshift
    \drawchord{1}{4}
    \drawchord{1}{3}
    \drawchord{5}{8}
    \drawchord{6}{8}
 \draw[gray] (.2,.2) node {$\bullet$};
\draw[gray] (.2,-.2) node {$\bullet$};
\draw[gray] (-.2,-.2) node {$\bullet$};
\draw[gray] (-.2,.2) node {$\bullet$};
\draw[gray] (-1,.2) node {$\bullet$};	
\draw[gray] (-.5,1.1) node {$\bullet$};
\draw[gray] (.8,-.3) node {$\bullet$};	
\draw[gray] (1.1,-.5) node {$\bullet$};
\draw[gray](1.1,-.5) -- (.8,-.3) --(.2,-.2); 
\draw[gray] (-.2, .2) --(-1,.2) -- (-.5,1.1);
\foreach \i in {1,2,...,8} {
      \draw[gray] (45*\i:2) node {$\bullet$};
    }
\draw[gray] (45*1:2) --(.2,.2);
\draw[gray] (-.2, -.2) -- (45*5:2);
\draw[gray] (45*6:2) -- (.8,-.3);
\draw[gray] (45*8:2) -- (1.1,-.5) -- (45*7:2);
\draw[gray] (45*2:2) -- (-.5,1.1) -- (45*3:2);
\draw[gray] (45*4:2) -- (-1,.2);
	\end{scope}
	\end{tikzpicture}
\eas

\caption{The polygon dissection $\tau(W_3)$ overlaid with its dual graph $T(W_3)$ (left diagram), and the forest obtained by splitting at the only vertex of degree $>3$ (right diagram).}
\label{fig:vertex splitting in the dual graph}
\end{figure}
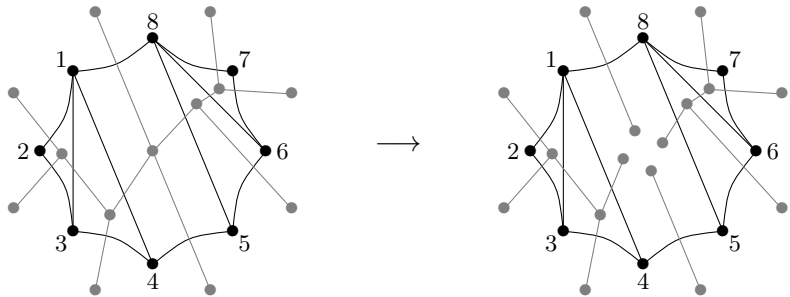

To see uniqueness, consider a different maximal decomposition of $\tau(W)$. This induces a splitting on $T(W)$, where each connected component of the new decomposition corresponds to a subtree. Call this forest $f'$. Since $f' \neq f$, there are two trees $t$ and $t'$ in $f$ and $f'$ respectively that are distinct, but share at least one edge of $T(W)$. Since $f'$ is also maximal, $t'$ is not a subtree of $t$. Therefore, the edges of $t'$ can be found in at least two trees in the forest $f$. In particular, there is a non-leaf vertex $v$ in $t'$ that corresponds to a split vertex of $T(W)$ in the original decomposition.
This implies that $v$ has degree greater than $3$ in $T(W)$, and thus the corresponding face of $\tau(W)$ is not a triangle. In other words, the piece of $\tau(W)$ corresponding to $f'$ is not a triangulation. 
\end{proof}

\begin{cor} \label{maxtriangdisjointcor}
Given a maximal decomposition of $\tau(W)$, the maximal triangulated pieces are edge disjoint.
\end{cor}

\begin{proof}
Consider any two distinct maximal triangulated pieces of $\tau(W)$. Using notation as in the previous proof, these two pieces correspond to subtrees of $T(W)$ and intersect, at most, at a vertex in the interior of $\tau(W)$. Since the subtrees corresponding to the maximal triangulated pieces are edge disjoint, and the edges of $T(W)$ correspond to the edges of $\tau(W)$, this forces the maximal triangulated pieces to be edge disjoint as well.
\end{proof}

%It is worth remarking that these maximal decompositions correspond to non-crossing partitions of the Wilson loop diagram, in the sense discussed in \cite{positroids13}. 

Next, we introduce the idea of retriangulations. 
\begin{dfn} \label{retriangulation}
  If $W = (\cP, [n])$ and $W' = (\cP', [n])$ are weakly admissible Wilson loop diagrams, then $\tau(W)$ and $\tau(W')$  differ by a \emph{retriangulation} if there is a bijection $f$ between the maximal triangulated pieces of $\tau(W)$ and the maximal triangulated pieces of $\tau(W')$ such that for any maximal triangulated piece $t$ of $\tau(W)$ the vertex sets of $t$ and $f(t)$ are the same.
\end{dfn}
More casually this definition says that the maximal triangulated pieces of $\tau(W)$ and $\tau(W')$ are the same so far as their vertex sets are concerned (and hence the regions triangulated by each piece are the same) but that the configuration of triangles within some or all of the triangulated pieces may be different. This definition of retriangulation makes sense for any two polygon dissections of the same polygon, not just for those in the image of $\tau$. 
   
We are now in a position to relate the triangulated pieces of $\tau(W)$ to exact subdiagrams of $W$.

For a weakly admissible Wilson loop diagram $W$, every triangulated piece $t$ of $\tau(W)$ corresponds to a subdiagram of $W$ obtained by taking the set of propagators $P$ corresponding to edges of $t$ and then taking the subdiagram $(P, V(P))$.  Conversely, given a subdiagram $(P, V(P))$ of $W$ we can obtain a subgraph of $\tau(W)$, called $t$, following the algorithm below.
\begin{itemize}
\item The vertex set of the subgraph $t$ is
  \begin{itemize}
  \item the vertices of $\tau(W)$ corresponding to edges of $W$ whose end points are elements of $V(P)$ which are cyclically consecutive in $[n]$. (This is significantly easier to interpret in picture form: see Example~\ref{eg:subtrees} below.)
  \end{itemize}
\item The edge set of the subgraph is
  \begin{itemize}
  \item the edges of $\tau(W)$ corresponding to propagators of $P$,
  \item the outer edges of $\tau(W)$ for which both their end points are in the vertex set.
  \end{itemize}
\end{itemize}

\begin{eg}\label{eg:subtrees}
Consider the following Wilson loop diagram $W$ and its corresponding polygon dissection $\tau(W)$. 
\bas W = \begin{tikzpicture}[rotate=67.5,baseline=(current bounding box.east)]
	\begin{scope}
	\drawWLD{10}{1.5}
	\drawnumbers
	\drawprop{1}{0}{4}{0}
	\drawprop{2}{0}{4}{-1}
	\drawprop{4}{1}{8}{0}
	\drawprop{8}{1}{1}{-1}
	\drawprop{5}{0}{7}{0}
%\draw (0,-2) --(0,2); horiontal
%\draw (-2,0) --(2,0);
%\draw (0,-2) node {\footnotesize $-y$};
%\draw (-2,0) node {\footnotesize $-x$};
\draw (2.7*\angle:0.8*\radius) node {\footnotesize $p$};
\draw (3.3*\angle:0.45*\radius) node {\footnotesize $q$};
\draw (0,-.4) node {\footnotesize $r$};
\draw (.7, 0) node {\footnotesize $s$};
		\end{scope}
	\end{tikzpicture}  \; \qquad
\tau(W) = \begin{tikzpicture}[rotate=67.5,baseline=(current bounding box.east)]
	\begin{scope}
	\drawpolypart{10}{1.5}
    \drawnumbersshift
    \drawchord{1}{4}
    \drawchord{2}{4}
    \drawchord{4}{8}
    \drawchord{8}{1}
    \drawchord{5}{7}
\draw (2.7*\angle:0.8*\radius) node {\footnotesize $p$};
\draw (3.3*\angle:0.47*\radius) node {\footnotesize $q$};
\draw (0,-.4) node {\footnotesize $r$};
\draw (.7, 0) node {\footnotesize $s$};
	\end{scope}
	\end{tikzpicture} \eas

The set of propagators $P= \{p,q,r,s \}$ defines an exact subdiagram of $W$, with $V(P) = \{1,2,3,4,5,8,9\}$.  The corresponding subgraph $t$ of $\tau(W)$ has vertex set $\{1, 2,3, 4, 8\}$. Note that we include vertex 3 of $\tau(W)$ in $t$ because the two endpoints of edge 3 in $W$ (namely the vertices 3 and 4 in $W$) are both in $V(P)$. However, we don't include vertex 5 of $\tau(W)$ in $t$ because only one of the endpoints of edge 5 in $W$ is in $V(P)$. The same argument applies to vertices 7, 9 and 10. Vertex 6 is excluded because none of its endpoints lie in $V(P)$. 

The subgraph $t$ is then 

\bas t = 
\begin{tikzpicture}[rotate=67.5,baseline=(current bounding box.east)]
	\begin{scope}
\pgfmathsetmacro{\n}{10}
\pgfmathsetmacro{\radius}{1.5}
\pgfmathsetmacro{\angle}{360/\n}
    \foreach \i in {1,2,3,4} {
      \draw (\angle*\i+ \angle/2:\radius) node {$\bullet$}; }
    \foreach \i in {2,3,4} {
     \pgfmathsetmacro{\x}{\angle*\i - \angle/2}
      \pgfmathsetmacro{\concave}{((\n-1.5)/\n)}
      \draw (\x:\radius cm) .. controls (\angle *\i: \concave* \radius cm) .. (\x + \angle:\radius cm);
      \draw (\angle*8+ \angle/2:\radius) node {$\bullet$};
}    
\foreach \i in {1,2,3,4,8} {
  \pgfmathsetmacro{\x}{\angle*\i + \angle/2}
  \draw (\x:\radius*1.15) node {\footnotesize \i};
}
    \drawchord{1}{4}
    \drawchord{2}{4}
    \drawchord{4}{8}
    \drawchord{8}{1}
	\end{scope}
	\end{tikzpicture} \eas

The key point to note is that the subgraph $t$ depends on how the propagators of $P$ sit inside $[n]$, not on how they sit in the subdiagram $(P,V(P))$ (whose corresponding graph would have additional outer edges between vertices 8 and 9, 9 and 1, 4 and 5, and 5 and 8).  In particular, $t$ is not the subgraph of $\tau(W)$ consisting only of edges corresponding to propagators of $P$ (as it may also include some of the outer edges of $\tau(W)$), nor is it the polygon dissection of $(P,V(P))$ viewed as a Wilson loop diagram in its own right (which would have vertex set equal to $V(P)$).
\end{eg}

\begin{lem}\label{lem triang to exact}
  Let $W$ be a weakly admissible Wilson loop diagram and $\tau(W)$ its polygon dissection.  The triangulated pieces of $\tau(W)$ correspond to the exact subdiagrams of $W$, via the correspondence described above.
\end{lem}

\begin{proof}
Let us first record a few elementary facts about polygon triangulations (that is, about triangulations with all vertices on the outer face).  If such a triangulation has $n$ vertices then it has $n$ edges on the polygon (that is, on the outer face) and $n-3$ edges which are not.  No planar simple graph with the same vertices and the same outer face can have more edges than a triangulation, and every such simple graph with $n-3$ edges off the outer face is a triangulation.\footnote{The last sentence is the statement that maximal outerplanar graphs are exactly polygon triangulations (a standard fact that possibly was first written down in \cite{MPhD}, see Theorem 1 of \cite{LMN}), then the earlier facts follow from Euler's formula for planar graphs.}

Since $W$ is weakly admissible, by Lemma \ref{tausimpleplanarlem} $\tau(W)$ is a simple graph. Let $t$ be a triangulated piece of the decomposition of $\tau(W)$ given in Lemma \ref{decompositionlem}.  Note that $t$ cannot be equal to $\tau(W)$ by the definition of admissible diagrams.

If $t$ has 2 vertices then $t$ corresponds to a propagator that connects two non-adjacent edges. Therefore, the trivial triangulation is a trivial exact subdiagram.

Now suppose that $t$ has $m>2$ vertices.  We count how many edges of $t$ are not on the outer face of $\tau(W)$.  These are exactly the edges of $t$ defined by propagators of $W$. Consider the intersection of $t$ with the outer face of $\tau(W)$: this is a possibly disconnected subgraph of the polygon of $\tau(W)$ and this subgraph has $m$ vertices. Call this subgraph $S$, and let $j$ be the number of connected components of $S$.   To join the components of $S$ into the outer face of $t$, $t$ must have $j$ edges in its outer face which are not in the outer face of $\tau(W)$.  Furthermore $t$ has $m-3$ edges not in its outer face and so also not in the outer face of $\tau(W)$.  Thus there are $m-3+j$ edges of $t$ not in the outer face of $\tau(W)$.

Each of these $m-3+j$ edges corresponds to a propagator in $W$.  Call this set of propagators $P$.  Next we count the size of $V(P)$.  Each of the $m$ vertices in the outer face of $t$ corresponds to an edge of $W$. These $m$ edges define $j$ connected components of the outer polygon of $W$. Thus the set $V(P)$ has $m+j$ vertices.  In other words, 
\[|V(P)| = m+j = |P| +3\;.\]
Thus the subdiagram $(P,V(P))$ defined by $t$ is exact.

Conversely, suppose we have an exact subdiagram $(P, V(P))$ of $W = (\cP, [n])$ supported on $|V(P)| = |P|+3$ vertices, and let $t$ be the subgraph of $\tau(W)$ corresponding to $(P,V(P))$.

Suppose $|P|=1$, and let $p$ be the single propagator belonging to $P$. The exactness condition on $(P, V(P))$ says that the four supporting vertices of $p$ are distinct.  If the support of $p$ is four consecutive vertices, then $V(p)$ defines three consecutive boundary edges of $W$, so $t$ is a single triangle, hence a triangulated piece.  If the support of $p$ is not four consecutive vertices, then the vertices which are the ends of $t$ are separated by at least two vertices along the cycle.  This implies that $t$ is a trivial triangulated piece.

Now suppose $|P|>1$.  Let $j=|P|$, $m= j+3=|V(P)|$, and 
suppose that the set $V(P)$ defines $c$ disjoint cyclic intervals of $[n]$. Then $t$ has $m-c$ vertices.  If $t$ were a triangulation, then $t$ would have $j -c$ internal edges, so we calculate the number of internal edges of $t$.

The graph $t$ has $j$ edges that come from propagators, and $m - 2c$ edges that come from the boundary polygon of $\tau(W)$. Since $t$ has $m-c$ vertices, it has $m-c$ external edges, of which $c$ come from propagators. Therefore, of the $j$ edges of $t$ that come from propagators, $j-c$ are internal to the connected component. Therefore, $t$ is a triangulated piece.

%  The statement of the lemma then follows form the fact that inclusion is preserved under $\tau$ and so maximality also corresponds under $\tau$.
\end{proof}

To avoid the issue of exact diagrams being subdiagrams of other exact subdiagrams (for instance, any subdiagram $(q, V(q))$ for $q \in \cP$ is exact), we introduce the notion of maximal exact subdiagrams.

\begin{dfn}
An exact subdiagram $(P, V(P))$ is a {\em maximal exact subdiagram} of $W$ if there is no other exact subdiagram $(Q, V(Q))$ in $W$ that contains $(P,V(P))$ as a strict subdiagram.
\end{dfn}

\begin{cor} \label{uniqueproppartitioncor}
Any weakly admissible Wilson loop diagram $W = (\cP, [n])$ can be uniquely decomposed into maximal exact subdiagrams. These maximal subdiagrams partition $\cP$ and correspond to the maximal triangulated pieces of $\tau(W)$.
\end{cor}

\begin{proof}
Combining Lemmas \ref{decompositionlem} and \ref{lem triang to exact} yields the unique decomposition into maximal exact subdiagrams, and Corollary \ref{maxtriangdisjointcor} ensures that no propagator appears in more than one subdiagram in this decomposition. Since the diagonal edges of $\tau(W)$ correspond to the propagators of $W$, the decomposition of $\tau(W)$ induces a partition of $\cP$.
\end{proof}

Finally, we may express equivalent Wilson loop diagrams in terms of retriangulations.

\begin{cor}\label{equialentretriangulation}
Two Wilson loop diagrams $W$ and $W'$ are equivalent if and only if $\tau(W)$ and $\tau(W')$ differ by retriangulations.
\end{cor}

\begin{proof}
This follows from the definition of equivalent Wilson loop diagrams (Definition \ref{equivdfn}), the definition of retriangulations (Definition \ref{retriangulation}), and Lemma \ref{lem triang to exact}.
\end{proof}

\subsection{Matroidal properties of exact subdiagrams \label{sec: exact diagram matroidal props}}

Since Corollary \ref{uniqueproppartitioncor} allows us to decompose any admissible Wilson loop diagram into a collection of maximal exact subdiagrams, we are motivated to examine the matroid properties of exact subdiagrams more closely. In this section, we prove two main results.
\begin{enumerate}
\item the matroid associated to an exact subdiagram of $W$ can be written as a contraction of the matroid $M(W)$ by the complementary propagator flat (Theorem \ref{exact diagrams contractions}),
\item the matroid associated to an exact subdiagram is uniform (Theorem \ref{exactuniformthm}).
\end{enumerate}

We begin by proving some useful facts about propagator flats, and flats of matroids associated to admissible Wilson loop diagrams more generally. Recall from Definition~\ref{VPropdfn} that the propagator flat $F(P)$ of $P \subseteq \cP$ denotes the set $V(P^c)^c$ of vertices that only support propagators in $P$ or are non-supporting.

\begin{lem} \label{lem decompose flat}Let $F$ be a flat in $M(W)$, and let $C \subseteq F$ be the union of all circuits contained in $F$. Then the following are true:
\begin{enumerate}
\item $C = F(\Prop (C))$, i.e. $C$ is a propagator flat,
\item $F \setminus C$ is an independent set. Furthermore, $F\setminus C$ is an independent flat if and only if $F(\emptyset) = \emptyset$, that is, $W$ has no non-supporting vertices.  
\end{enumerate}
\end{lem}

\begin{proof}
(1) If $F$ is an independent flat, then $C = \emptyset$ and the statement is trivially true. Now suppose that $F$ is a dependent set, so $C$ is non-empty.

Let $v \in C$. Clearly $\Prop(v) \subseteq \Prop(C)$, and hence $F(\Prop(v)) \subseteq F(\Prop(C))$ (this is easily verified by applying the definition and simplifying out all of the complements). Since $v \in F(\Prop(v))$ by the definition of propagator flat, it follows that $v \in F(\Prop(C))$. Thus $C \subseteq F(\Prop(C))$.

Now suppose there exists some $w \in  F(\Prop(C)) \setminus C$. Let $B$ be an independent subset of $C$ of maximal rank. We first show that $B \cup \{w\}$ is a dependent set in $M(W)$. By Lemma \ref{lem facts about WLD matroids}, we have
\begin{equation}\label{eq:equality in lem decompose flat} |\Prop(C)| = \rk(C) = |B| = \rk(B) .\end{equation}
Note that $\rk(B) \leq |\Prop(B)|$, and $B \subseteq C$ implies that $\Prop(B) \subseteq \Prop(C)$. Combining these facts with equation \eqref{eq:equality in lem decompose flat}, it follows that $\Prop(B) = \Prop(C)$. Combining this with the fact that $w \in F(\Prop(C))$, we have $\Prop(B\cup \{w\}) \subseteq \Prop(B)$. We can now say something about $\rk(B\cup \{w\})$: by Lemma \ref{lem facts about WLD matroids}, we have 
\[
\rk (B \cup \{w\}) \leq \min\{|B\cup \{w\}|, |\Prop(B\cup \{w\})|\} \leq \min\{|B|+1,|\Prop(B)|\},
\]
which becomes 
\[\rk(B\cup \{w\}) \leq \min\{|B|+1,|B|\} = |B| < |B \cup\{w\}|\]
via equation \eqref{eq:equality in lem decompose flat}. Therefore $B\cup\{w\}$ is a dependent set, which implies that $B\cup \{w\}$ contains a circuit $D$. In particular we must have $w \in D$ (since $B$ was an independent set), and since we have already shown that $C \subseteq F(\Prop(C))$, we also have the chain of inclusions 
\[D \subseteq B\cup \{w\} \subseteq F(\Prop(C)) \subseteq F.\]
Since $C$ was the union of all circuits in $F$, we finally obtain $w \in C$, which is a contradiction. This completes the proof of (1).

For part (2), first note that $F \setminus C$ is automatically independent as it contains no circuits.

Note that since $F(\emptyset)$ has rank $0$, $F(\emptyset)$ is contained in every flat of $W$. Therefore, if $F(\emptyset) \neq \emptyset$, then $F\setminus C$ cannot be a flat. 

Now suppose that $F(\emptyset) = \emptyset$. We must show that $F\setminus C$ is a flat. For any $e \not\in F$ we certainly have $\rk((F\setminus C)\cup\{e\}) = \rk(F\setminus C) +1$, since $F$ is a flat. Now let $e \in C$, and suppose that $\rk((F\setminus C)\cup\{e\}) = \rk(F\setminus C)$. This implies that $(F\setminus C)\cup \{e\}$ is dependent, and hence contains a circuit. Since $F(\emptyset) = \emptyset$ this circuit must contain at least two elements: one of these elements is $e$, but any others must come from $F \setminus C$. This contradicts the fact that $C$ was the union of all circuits in $F$. 

Thus $\rk((F\setminus C)\cup \{e\}) = \rk(F\setminus C) + 1$ for any $e \not\in F\setminus C$, and hence $F\setminus C$ is a flat.
\end{proof}

\begin{cor} \label{classifyflats}
  If $F$ is a flat of a weakly admissible Wilson loop diagram, it can be written as the disjoint union of a cyclic propagator flat $C$ and an independent set $S$: \bas F =C \sqcup S \eas where $C$ the union of all circuits in $F$. \end{cor}
 
\begin{proof}
By Lemma \ref{lem decompose flat} the set $C = F(\Prop(C))$ is a propagator flat, and $S := F \setminus C$ is independent.
\end{proof}

In particular, any propagator flat can be written as a union of a cyclic propagator flat and an independent set.

Next we examine properties of propagator flats associated to the complement of Wilson Loop diagrams.
%Finally, we note the differences between subdiagrams of Wilson loop diagrams, and restrictions or cotractions of the associated matroids.

%Given any matroid, one may restrict it to a subset of the base set. The bases of the restricted matroid come from intersecting bases of the original with the subset. It is worth noting that a the matroid defined by a subdiagram is different from the restriction of the matroid of a Wilson loop diagram to a set of vertices.

% \begin{dfn} \label{restrictiondfn}
% For $W = (\cP, [n])$, the restricted diagram, $W|_V$ is the matroid defined by only looking at the vertices $V \subset [n]$.
% \end{dfn}

% The key difference between a subdiagram and a restriction is that the propagator support function does not change in the case of restriction, while it may in the case of a subdiagram. In particular, for $v \in V$, $\Prop_W(v) = \Prop_{W|_V}(v)$, while $\Prop_{(P, V(P))} (v) = \Prop_W(v) \cap P$.

% A subdiagram is more closely related to a contracted matroid.

\begin{lem} \label{maxexactcomplementrank}
Let $W = (\cP, [n])$ be an admissible Wilson loop diagram, and $P \subseteq \cP$. Then, \begin{enumerate}
\item if $(P,V(P))$ is an exact subdiagram of $W$, then $F(P^c)$ is a dependent flat in $W$,
\item if $(P,V(P))$ is a maximal exact subdiagram of $W$, then $F(P^c)$ is a cyclic flat in $W$,
\item if $(P,V(P))$ is an exact subdiagram of $W$, then $\rk(F(P^c)) = |P^c|$.
\end{enumerate} 
\end{lem}

\begin{proof}
%Note that $|P^c|$ is always an upper bound on the rank of $F(P^c)$ for any $P \subseteq \cP$, since $F(P^c)$ supports at most $P^c$ propagators (see Remark \ref{alt F(P) rmk}), and $\rk (F(P^c)) = \min\{ |F(P^c)|, |\Prop(F(P^c))|\}$. To show equality, we need to show that $F(P^c)$ is not an independent set and that $\Prop(F(P^c)) = P^c$. 

First note that if $(P,V(P))$ is exact then the admissibility of $W$ guarantees that $V(P) \neq [n]$, so $F(P^c) = V(P)^c$ is not empty. 

We know from Lemma~\ref{lem facts about WLD matroids} that $F(P^c)$ is always a flat, so for item 1 we only need to show that $F(P^c)$ is dependent.

Since $W$ is admissible, we have $n \geq |\cP| + 4$. Rewriting this inequality as
\[|V(P)| + |F(P^c)|  \geq  |P| + |P^c| +4,\]
and combining it with the fact that $|V(P)| = |P| + 3$ (from the exactness of $(P,V(P))$), we obtain
\begin{equation}\label{eq F not independent}|F(P^c)| > |P^c| \;.\end{equation}
Equation \eqref{eq F not independent} implies that $F(P^c)$ supports fewer propagators than the number of vertices it contains, i.e. $F(P^c)$ is not an independent set.

For item 2, suppose that $(P, V(P))$ is a maximal exact subdiagram of $W$. By Lemma \ref{lem decompose flat} item 1 and Corollary \ref{classifyflats} we can decompose $F(P^c)$ as
\begin{equation}\label{eq decompose F}V(P)^c = F(P^c) = C \sqcup S = F(\Prop(C)) \sqcup S,\end{equation}
where $C$ is the largest cyclic flat contained in $F(P^c)$ and $S$ is an independent set. We show that $\Prop(C) = P^c$, thus forcing $S$ to be empty and implying that $F(P^c) = C$, i.e. $F(P^c)$ is a cyclic flat. Note that $C$ must be non-empty since $F(P^c)$ is non-empty and dependent (by item 1). 

By similar arguments to Lemma~\ref{lem decompose flat}, the fact that $S$ is independent and that $C$ is the union of all circuits in $F(P^c)$ implies that every element of $S$ is independent of $C$. Further, we know that $\rk(C) = |\Prop(C)|$ from Lemma \ref{lem facts about WLD matroids}(2), and $\rk(S) = |S|$ from Lemma \ref{lem facts about WLD matroids} item 1 and the independence of $S$.  Combining these gives
\ba \rk\big(F(P^c)\big) = \rk(C) + \rk(S)  = |\Prop(C)| + |S| \leq |\Prop(F(P^c))|\;,\label{rank of comp flat} \ea 
where the inequality follows from the bound on ranks of vertex sets in Lemma \ref{lem facts about WLD matroids} item 1. For ease of future notation, write $Q = \Prop(C)^c$. 

From Remark \ref{alt F(P) rmk} it follows that $\Prop (F(P^c)) \subseteq P^c$, and so equation \eqref{rank of comp flat} implies \bas |Q^c| + |S| \leq |P^c|  \;.\eas 

Rearranging and adding $|P|$ to both sides gives \ba |P| + |S| \leq |P| + |P^c| - |Q^c| = |Q| \;.\label{P+Sbound}\ea 
Meanwhile, by taking complements of equation \eqref{eq decompose F} it follows that
\bas V(P) \sqcup S = C^c = F(\Prop(C))^c  = V(Q)\;. \eas
Combining this with equation \eqref{P+Sbound} gives 
\ba |V(Q)| = | V(P) | + |S| = | P | + 3 + |S| \leq |Q| +3 \;. \label{Q defines exact}\ea
Since $W$ is admissible this cannot be a strict inequality, so we have $|V(Q)| = |Q| +3$ and the subdiagram $(Q, V(Q))$ is exact.

Since $\Prop(C) \subseteq \Prop(F(P^c)) \subseteq P^c$, it follows that $P \subseteq \Prop(C)^c = Q$. In other words, $(P, V(P))$ is a subdiagram of $(Q, V(Q))$. Since $(P, V(P))$ is a maximal exact subdiagram by assumption, we must have $P=Q$. Recalling that $Q = \Prop(C)^c$, we finally obtain $\Prop (C) = P^c$ and hence $S = \emptyset$. This completes the proof of item 2.

For item 3, first note that when $(P,V(P))$ is maximal exact, item 2 of this Lemma combined with Lemma \ref{lem facts about WLD matroids} gives that
\bas\rk(F(P^c)) = |\Prop(F(P^c))| = |P^c|.\eas

To prove the result in the general case, let $(R,V(R))$ be an exact subdiagram of $W$ that is not maximal and let $(P, V(P))$ be the maximal exact subdiagram containing it. That is, $R \subsetneq P$.

Since $R \subsetneq P$ we can write 
\begin{equation}\label{eq: relate P and R diagrams}V(P) = V(R) \sqcup S, \quad \text{ where } S := V(P) \setminus V(R).\end{equation}
Since $P$ and $R$ both define exact subdiagrams, it follows that $|V(P)| = |V(R)| + |P \setminus R|$, and hence that $S$ is a vertex set of size $|P \setminus R|$. Furthermore, by taking complements of both sides of equation~\eqref{eq: relate P and R diagrams}, we see that $F(R^c) = F(P^c) \sqcup S$. 

We will show that $S$ is an independent set in the subdiagram $(P, V(P))$, and thus in $W$. This will imply that for any maximal independent set $T \subseteq F(P^c)$, the set $T \sqcup S$ is independent in $F(R^c)$, and hence that 
\bas \rk (F(R^c)) \geq |T| + |S| = \rk (F(P^c)) + |S| = |P^c| + |P\setminus R| = |R^c|  \;.\eas 
The reverse inequality $\rk (F(R^c)) \leq |R^c|$ always holds (see Lemma \ref{lem facts about WLD matroids}), so if $S$ is independent in $(P,V(P))$ we obtain $\rk (F(R^c)) = |R^c|$ as desired.

It remains to show that $S$ is an independent set in the subdiagram $(P, V(P))$.  By way of contradiction suppose otherwise. Then by Theorem \ref{thm WLD defines matroid} there is a subset $U \subseteq S$ such that $|\Prop(U)| < |U|$. Then the remaining propagators $P \setminus \Prop(U)$ are supported entirely on $V(P) \setminus U$, i.e. \bas  V(P \setminus \Prop(U)) \subseteq V(P) \setminus U\; .\eas Comparing sizes of these sets, we get \bas |V(P \setminus \Prop(U))| \leq |V(P)| - |U| = |P| +3 -|U| < |P| +3 -|\Prop(U)| = |P \setminus \Prop(U)| +3\; ,\eas which violates admissibility. Therefore $S$ must be independent in $(P, V(P))$ after all, and (as described above) it follows that $\rk(F(R^c)) = |R^c|$.

\end{proof}

We are now in a position to describe the matroid structure of exact subdiagrams $(P, V(P))$ of $W$ in terms of $M(W)$ and $F(P^c)$. We begin with a definition.

\begin{dfn}\label{matroid contraction}
Let $M = (E,\cB)$ be a matroid, and $S \subseteq E$. The {\em contraction} of $M$ by $S$ is the matroid $M/S = (E \setminus S, \cB / S)$, where
\[\cB / S = \{B \setminus S \ \big| \ |B\cap S | \text{ is maximal amongst all }B \in \cB\}.\]
\end{dfn}

In \cite{wilsonloop}, Agarwala and Marin-Amat show that certain subdiagrams of $W$ can be realized as contractions of $M(W)$, which we rephrase using the notation of this paper:

\begin{lem} \label{contractsubdiaglem} \cite[Corollary 3.33]{wilsonloop} 
Let $W = (\cP, [n])$ be a Wilson loop diagram and $P \subseteq \cP$. The set $V(P)^c$ has rank $|P^c|$ (i.e. $F(P^c)$ is a propagator flat of maximal rank), if and only if the matroid defined by the subdiagram $(P, V(P))$ is equal to the contraction $M(W)/V(P)^c$.
\end{lem}

Putting all of this together, we obtain our first major theorem of this section:

\begin{thm} \label{exact diagrams contractions}
If $(P, V(P))$ is an exact subdiagram of $W$, the matroid of $(P,V(P))$ is the contraction of $M(W)$ by the propagator flat $F(P^c)$, i.e. \bas M\big((P, V(P))\big) = M(W)  / F(P^c) \;.\eas 
\end{thm}

\begin{proof}
This follows from Lemma \ref{contractsubdiaglem}. Lemma \ref{maxexactcomplementrank} shows that the supports of exact subdiagrams satisfy the conditions of this lemma.
\end{proof}

\begin{rmk} \label{remark exact dual restiction} For readers familiar with duals and complements of matroids, we reinterpret Theorem~\ref{exact diagrams contractions} in these terms. If $M^*$ denotes the dual of a matroid, then by standard matroid theory \cite{OxleyMatroidBook} the dual of a contraction of a matroid is the same as the restriction of the dual matroid by the complement: \bas M / S = M^*|_{S^c} \;.\eas Then Theorem \ref{exact diagrams contractions} implies that, for $(P, V(P))$ an exact subdiagram of $W$, we have \bas M\big((P, V(P))\big) = M(W)  / F(P^c) = M(W)^*|_{V(P)} \;.\eas
\end{rmk}

We now show that matroids coming from exact subdiagrams have an especially nice structure, namely, they are uniform. Recall from Section \ref{sec matroid background} that a uniform matroid of rank $r$ is a matroid in which all sets of size $ \leq r$ are independent.

\begin{thm} \label{exactuniformthm}
Let $W':= (P, V(P))$ be a subdiagram of a weakly admissible Wilson loop diagram $W= (\cP, [n])$. Then $W'$ is an exact subdiagram if and only if $M(W')$ is a uniform matroid of rank $|P|$.
\end{thm}

\begin{proof}
It follows directly from the definitions that a matroid of rank $r$ is uniform if and only if all circuits have rank $r$. We therefore focus on the circuits of $M(W')$.

We prove the following claim: $W'$ is exact if and only if $V(P)$ contains no circuits $C$ with $\rk (C)< |P|$ in $(P, V(P))$. Since $\rk(M(W'))$ is bounded above by $|P|$, this will give that  $M(W')$ is uniform.

First suppose $C \subseteq V(P)$ is a circuit of rank $m < |P|$. Observe that $\Prop_{W'}(C)\subseteq P$ by definition, where the subscript to $\Prop$ specifies the diagram we are working in. By Lemma \ref{lem facts about WLD matroids} we know that ${|\Prop_{W'}(C)| = m}$. Thus the set $P \setminus \Prop_{W'}(C)$ must be nonempty, and we can consider the subdiagram $W'':= (P\setminus \Prop_{W'}(C),V(P\setminus\Prop_{W'}(C))$. By the density condition on subdiagrams of admissible diagrams, we have
\[|V(P\setminus\Prop_{W'}(C))| \geq |P\setminus\Prop_{W'}(C)| + 3.\]
It is easy to verify that $V(P\setminus\Prop_{W'}(C)) \subseteq V(P)\setminus C$. Since $C \subseteq V(P)$ and $\Prop_{W'}(C) \subseteq P$, it follows from the previous inequality that
\[|V(P)| - (m+1) \geq |V(P\setminus\Prop_{W'}(C))| \geq |P| - m + 3.\]
Simplifying, we obtain $|V(P)| \geq |P| + 4$, i.e. $(P,V(P))$ is not an exact diagram.

Conversely, suppose that $W'$ is not exact and for a contradiction suppose also that $M(W')$ is uniform of rank $|P|$.  Take $p \in P$.  Then $|V(P)\setminus V(p)| = |V(P)| - 4 \geq |P|$ by non-exactness.  By uniformity there must be an independent set of size $|P|$ in $V(P)\setminus V(p)$.  However, this is impossible because the submatrix of $C(W')$ corresponding to this independent set has $|P|$ rows but the one corresponding to $p$ is all $0$, so it cannot be full rank.
%suppose that $(P,V(P))$ is not exact, i.e. $|V(P)| \geq |P| +4$. We have 
%\[|V(P)| - (m+1) \geq |P| - m + 3\]
%as above, but \hlfix{we need to know that $P\setminus\Prop(C)$ is non-empty before we can complete the argument.}{and if we knew this, it would be immediate that $\rk(C) < |P|$}
\end{proof}

To close this section, we make a few observations about the geometry of matroids defined by exact diagrams.

In \cite[Theorem 3.39]{wilsonloop}, the authors show that each weakly admissible Wilson loop diagram corresponds to a positroid. That is, each diagram corresponds to a matroid that can be represented by elements of the positive Grassmannian $\Gr(|\cP|, n)$. Any positroid of rank $k$ on $n$ elements defines a subspace of the positive Grassmannian $\Gr(k, n)$, namely the points which represent it \cite{Postnikov}. These subspaces give a CW structure on $\Gr(k,n)$, with each positroid defining a cell.  In other words, there is a map from weakly admissible Wilson loop diagrams $W$ to matroids $M(W)$ to positroid cells in the CW complex on $\Gr(|\cP|, n)$.

With this mapping in mind, we have the following corollary:

\begin{cor}
Let $W' = (P, V(P))$ be an exact subdiagram of $W$. The matroid associated to this subdiagram corresponds to the top dimensional cell in $\Gr(|P|, |V(P)|)$.
\end{cor}

\begin{proof}
Since $W'$ is weakly admissible, by Theorem \ref{thm WLD defines matroid} the associated matroid $M(W')$ can be realized by an element of the totally nonnegative Grassmannian $\Gr(|P|,|V(P)|)$. However, since $W'$ is an exact subdiagram its corresponding matroid $M(W')$ is uniform, so all $|P| \times |P|$ minors of $C(W')$ must be non-zero and $M(W')$ can be represented by a matrix with all maximal minors strictly positive. Since the unique top dimensional cell of $\Gr(|P|, |V(P)|)$ is defined by precisely those points in $\Gr(|P|, |V(P)|)$ with the property that all Pl\"ucker coordinates are strictly greater than $0$, this completes the proof.
\end{proof}

\subsection{Matroids and equivalent diagrams \label{sec: matroids and equivalence}}

Now we are ready to prove the main result of this section, namely that two Wilson loop diagrams define the same matroid (positroid) if and only if they are equivalent (Theorem \ref{same matroid iff equiv}). We also give a formula for the number of Wilson loop diagrams in a given equivalence class (Corollary \ref{number of equiv diagrams}), completing the characterization of the correspondence between Wilson loop diagrams and positroids started in \cite{wilsonloop}.

\begin{thm}\label{same matroid iff equiv}
Let $W= (\cP, [n])$ and $W'= (\cP', [n])$ be two weakly admissible Wilson loop diagrams. They define the same matroid (positroid) if and only if $W \sim W'$.
\end{thm}

\begin{proof}One direction has been proved in \cite[Theorem 1.18]{wilsonloop}, but we give a different proof here to be consistent with the method of this paper. 

For the if direction, it suffices to prove the result for two equivalent diagrams that differ in exactly one exact subdiagram. Let $W$ and $W'$ be equivalent Wilson loop diagrams which can be decomposed as
\[W = (P \sqcup R, [n]) \qquad \textrm{and} \qquad W' = (P \sqcup R', [n]),\]
where $(R,V(R))$ and $(R',V(R'))$ are different maximal exact subdiagrams of $W$ and $W'$ respectively, and $V(R) = V(R')$. 

By Theorem \ref{exactuniformthm}, $M(R, V(R))$ and $M(R',V(R'))$ are both uniform matroids of rank $|R| = |R'|$. In particular, this means that any subset of $V(R)$ is independent in both $M(W)$ and $M(W')$.

Now consider the propagator flat $F(P)$, which is a cyclic flat of rank $|P|$ in both $M(W)$ and $M(W')$ by Lemma \ref{maxexactcomplementrank}. Let $B \subseteq F(P)$ be an independent set of rank $|P|$. Adding any element of $V(R)$ to $B$ increases its rank in both matroids since $F(P)$ is a flat and $F(P) \cap V(R) = \emptyset$, so we can lift any independent set in $F(P)$ to an independent set of both $M(W)$ and $M(W')$ by adjoining any independent set in $V(R)$. In particular, $B\sqcup U$ is a basis in both $M(W)$ and $M(W')$ for any $U \subseteq V(R)$ with $|U| = |R|$.

Finally, note that any basis of $M(W)$ or $M(W')$ must be of this form, since $[n] = F(P) \sqcup V(R)$ and $\rk F(P) + \rk V(R) = |P| + |R| = n$. Thus both matroids have the same set of bases, proving that they are the same.

For the converse, suppose that $M(W) = M(W') = M$. This immediately implies that $|\cP| = |\cP'|$, so to demonstrate the equivalence of $W$ and $W'$ it suffices to show that for any maximal exact subdiagram $(P,V(P))$ in $W$, there is a corresponding maximal exact subdiagram $(P',V(P'))$ in $W'$ with $V(P) = V(P')$.

So let $(P,V(P))$ be a maximal exact diagram in $W$. By Lemma~\ref{maxexactcomplementrank}, this implies that $F:=F(P^c)$ is a cyclic flat in $M$. Viewing $M$ as the matroid of $W'$, Lemma~\ref{lem decompose flat} implies that $F = F(\Prop_{W'}(F))$ (recall that ``cyclic'' simply means ``a union of circuits'', so the union of all circuits in $F$ is $F$ itself). Define $Q^c = \Prop_{W'}(F)$, so that $F(P^c) = F = F(Q^c)$, and therefore $V(P) = V(Q)$.

By Lemma~\ref{maxexactcomplementrank} we know that $\rk(F) = |P^c|$. However, by Lemma~\ref{lem facts about WLD matroids} we also know that the rank of a cycle is equal to the number of propagators supported on that cycle, so $\rk(F) = |\Prop_{W'}(F)| = |Q^c|$ as well. Since $(P,V(P))$ is exact and we have shown that $V(P) = V(Q)$ and $|P| = |Q|$, it follows that $|V(Q)| = |Q| + 3$, i.e.  $(Q,V(Q))$ is exact in $W'$.

However, $(Q,V(Q))$ is not a priori {\em maximal} exact, so let $(P',V(P'))$ be the maximal exact subdiagram in $W'$ that contains $(Q,V(Q))$ as a subdiagram (possibly trivially). Applying the argument above to $(P',V(P'))$, we obtain an exact subdiagram $(R,V(R))$ in $W$ with $V(R) = V(P')$ and $|R| = |P'|$. Now we have
\[V(P) = V(Q) \subseteq V(P') = V(R),\]
which implies that $(P,V(P)) = (R,V(R))$ by the maximality of $(P,V(P))$, and hence $(Q,V(Q)) = (P',V(P'))$ is maximal exact in $W'$.
\end{proof}

Since there is a unique way to decompose $W$ into maximal exact subdiagrams, it is logical to ask how many diagrams there are in an equivalence class. It is a classical fact (see for instance Chapter 6 of \cite{Stanley}) that the number of triangulations of a convex $n$-gon is the $n-2$ Catalan number, namely $\frac{1}{n-1}\binom{2(n-2)}{n-2}$.  Thus we can count the number of equivalent diagrams.

\begin{cor}\label{number of equiv diagrams}
  Let $W$ be a weakly admissible Wilson loop diagram where the sizes of the supports of the nontrivial maximal exact subdiagrams are $n_1, n_2, \ldots, n_j$.  Then the number of weakly admissible Wilson loop diagrams equivalent to $W$ (including $W$ itself) is
  \[
  \prod_{i=1}^{j} \frac{1}{n_i-1}\binom{2(n_i-2)}{n_i-2} \;.
  \]
\end{cor}

In this section we have characterized the correspondence between Wilson loop diagrams and positroids, showing that Wilson Loop diagrams define the same positroid if and only if they are equivalent. If we write \bas M: \textrm{admissible WLD} & \longrightarrow \textrm{Positroids} \\ W & \mapsto M(W) \;,\eas the image of this map is given by the distinct equivalence classes. Corollary \ref{number of equiv diagrams} enumerates the elements of each fiber of $M$.  The obvious next question would be to enumerate the number of equivalence classes. 

We know from examples (e.g. \cite[Remark 4.2]{casestudy}) that many positroid cells cannot be obtained as the matroids of admissible Wilson loop diagrams. We do not obtain a formula for the number of cells in the image of $M$ here, but in the next section we give an alternate characterization of ${\rm im}(M)$ in terms of associahedra, with the hope that this linkage with a well-studied combinatorial object will shed light on the issue in the future.

\section{Associahedra and inequivalent Wilson loop diagrams}\label{sec associahedron}

The correspondence between (weakly) admissible Wilson loop diagrams and polygon dissections shown in Section \ref{sec: polygon partitions} suggests a connection to associahedra. In order to explore this connection further, we first need some background, which is given in Section~\ref{sec polytope background}. In Section~\ref{sec associahedron results}, we show that the number of equivalence classes of Wilson loop diagrams correspond to the number of non-parallel faces of the appropriate associahedron (Theorem \ref{thm:count inequiv diagrams}). 

\subsection{Polytope background}\label{sec polytope background}

We first give a quick summary of the polytope terminology we will use in this section; the interested reader is referred to \cite{Ziegler} for more details.

A \emph{polytope} can equivalently be defined as the convex hull of a finite set of points in some $\mathbb{R}^d$, or as the intersection of finitely many closed half spaces in $\mathbb{R}^d$ with the additional property that this intersection is bounded in the sense of not containing any ray (see chapter 0 of \cite{Ziegler}).

A \emph{face} of a polytope is the intersection of the polytope with any hyperplane for which the entire polytope is in one of the closed half-spaces defined by the hyperplane.  The polytope itself is also considered a face. In the case that the polytope is of full dimension, this definition can be unified by considering linear equalities (including the trivial one) in place of hyperplanes and the corresponding linear inequalities in place of the half-spaces defined by the hyperplane (see Chapter 2 of \cite{Ziegler}).  

Two faces are \emph{parallel} if one can be mapped to the other by a translation in $\mathbb{R}^d$.  In particular parallel faces must be of the same dimension. Any two vertices are parallel to each other.

Finally, two polytopes are \emph{combinatorially equivalent} if there is a bijection between their faces which preserves inclusion of faces.  Often in combinatorics one is primarily interested in the combinatorial equivalence classes, and common polytope names (including the associahedron) refer to these equivalence classes.  A particular polytope in the class is then known as a \emph{realization}.  

As with many of the interesting and well-studied polytopes in combinatorics, the associahedron can be characterized in many different ways. Here is one definition.

\begin{dfn}\label{def:associahedron} 
Fix $n$ and consider dissections of a convex $n$-gon.  The $n-1$ {\em associahedron} can be defined in terms of the dissections of this $n$-gon: the vertices of the associahedron are the triangulations of the $n$-gon, while the $k$-dimensional faces are the dissections which are $k$ diagonals away from being a triangulation.  One face is included in another if the first dissection is a refinement of the second. 
\end{dfn}

The associahedron coming from dissections on an $n$-gon is known as the $n-1$ associahedron because this is the associahedron for which the vertices correspond to the ways to bracket a list of length $n-1$.  It is an $n-3$ dimensional polytope.

It turns out that this defines a class of combinatorially equivalent polytopes, which can be proved by giving a realization. There are many interesting realizations of the associahedron {\cite{CSZinequivalent}}. The following realization will be the one that we use.  

\begin{dfn}\label{def:secondary polytope} 
  Let $x_1, \ldots, x_n$ be the corners of a convex $n$-gon in $\mathbb{R}^2$, and let $T$ be the set of triangulations of this $n$-gon.  For each $t\in T$, define $s_t$ to be the point in $\mathbb{R}^n$ whose $i$th coordinate is the sum of the areas of all triangles of $t$ which are incident to $x_i$.  Let $A_n$ be the convex hull of the $s_t$. Then $A_n$ is a realization of the $n-1$ associahedron (\cite[Example 9.11]{Ziegler}).
\end{dfn}
Note that this construction works for any convex $n$-gon, but these realizations are all equally good for our purposes, and so we fix $x_1, \ldots, x_n$ defining a particular convex $n$-gon from now on.

Definition~\ref{def:secondary polytope} is a special case of the {\em secondary polytope} construction; see \cite[Definition 9.9]{Ziegler} for more details.  In general the $x_1, \ldots, x_n$ may be any points in any $\mathbb{R}^d$.  Triangulations in this more general context are decompositions into simplices, and the area of the triangles is replaced by the volume of the simplices.

\subsection{Non-parallel faces in the associahedron and inequivalent Wilson loop diagrams}\label{sec associahedron results}

The map $\tau$ defined in Definition \ref{WLDtriangulationdfn} assigns a dissection of an $n$-gon to a Wilson loop diagram $W = (\cP, [n])$. Recall that if $|\cP| = k$, then $\tau(W)$ is a dissection with $k$ diagonals. Therefore, the map $\tau$ associates to each weakly admissible Wilson loop diagram an $n-3 - k$ dimensional face of the associahedron $A_n$. This is illustrated for $n=5$ on the left in Figure~\ref{associahedron5gon}, and for $n=6$ (with some labels omitted for clarity) on the right.

\begin{figure}
\[\includegraphics{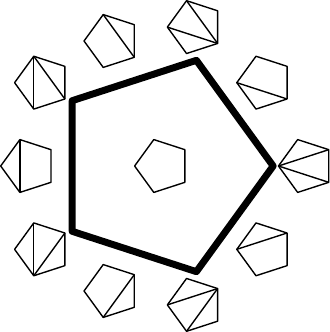} \hspace{1cm} \includegraphics{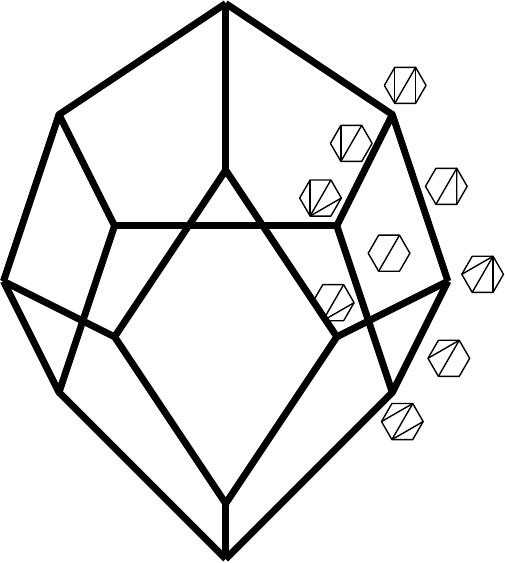}\]
\caption{Two associahedra, labelled by all and some (respectively) of their polygon dissections.}\label{associahedron5gon}
\end{figure}

Note that Wilson loop diagrams which are weakly admissible but not admissible are necessarily exact, and hence have $\tau$ graphs which are triangulations, i.e. dimension 0 faces (vertices) in $A_n$.  All of these are equivalent (as Wilson loop diagrams) in the sense that they are only retriangulations away from each other. Meanwhile, as noted above, all vertices of $A_n$ are parallel to each other.  It will be convenient in what follows to include these diagrams, but the results are easily modified to remove them if needed: Propositions~\ref{prop easy way} and \ref{prop hard way} would be identical if we replaced ``weakly admissible'' by ``admissible'', while Theorem~\ref{thm:count inequiv diagrams} could be changed to consider only admissible diagrams by subtracting 1 from the number of non-parallel faces in order to account for the missing dimension 0 faces.

\begin{prop}\label{prop easy way}
  If two weakly admissible Wilson loop diagrams are equivalent then their faces in $A_n$ are parallel.
\end{prop}

\begin{proof}
  Let $W$ and $X$ be two equivalent Wilson loop diagrams.  From Corollary \ref{equialentretriangulation} we know that $\tau(W)$ and $\tau(X)$ differ from each other only by retriangulations of their triangulated pieces. 

All edges of $A_n$ are polygon dissections which are only one diagonal away from being triangulations.  Such dissections consist of triangles and exactly one quadrilateral.  Consider the edges of $A_n$ bounding the faces corresponding to $\tau(W)$ and $\tau(X)$.  These edges have their one quadrilateral within the nontriangulated parts of $\tau(W)$ or $\tau(X)$, but these nontriangulated parts are the same in $\tau(W)$ and $\tau(X)$, so we get a bijection between the edges of $A_n$ bounding the face of $\tau(W)$ and the edges of $A_n$ bounding the face of $\tau(X)$ which matches edges with the same quadrilateral.

Consider now an edge in $A_n$ whose quadrilateral is defined by the vertices of $\tau(W)$ or $\tau(X)$, which we label $x_i$, $x_j$, $x_k$, and $x_l$ in that order (see Figure \ref{triangulate} below).  Let $t_1$ and $t_2$ be the triangulations corresponding to the two ends of the edge where $t_1$ has a diagonal from $x_j$ to $x_l$ and $t_2$ has a diagonal from $x_i$ to $x_k$.  From the construction of the secondary polytope $A_n$ in \ref{def:secondary polytope}, we see that $s_{t_1}$ and $s_{t_2}$ only differ at coordinates $i$, $j$, $k$, and $l$: all of the triangles incident to any of these vertices other than the one coming from triangulating $x_i, x_j, x_k, x_l$ are the same in both $t_1$ and $t_2$.  

Let $a$ be the area of the quadrilateral and let $a(i,j,k)$ be the area of the triangle with corners $x_i$, $x_j$ and $x_k$ (and similarly for other indices), and let $b_i$ denote the sum of the areas of triangles incident to vertex $x_i$ which are common to both triangulations. Then $s_{t_1}$ on coordinates $i$, $j$, $k$, $l$ is 
\[(a(i,j,l) + b_i, a+b_j, a(k,l,j)+b_k, a+b_l),\]
and $s_{t_2}$ on these coordinates is 
\[(a + b_i, a(j,k,i)+b_j, a+b_k, a(l,i,k)+b_l).\] 
This is illustrated in Figure~\ref{triangulate}.

Thus a vector parallel to this edge is the vector which is 
\begin{equation}\label{eq:direction vector for quadrilateral}(a(i,j,l) - a, a-a(j,k,i), a(k,l,j)-a, a-a(l,i,k)) = (-a(j,k,l), a(i,k,l), -a(i,j,l), a(i,j,k))\end{equation}
on $i$, $j$, $k$, $l$, and $0$ in all other coordinates. 

\begin{figure}[t]
\[\includegraphics{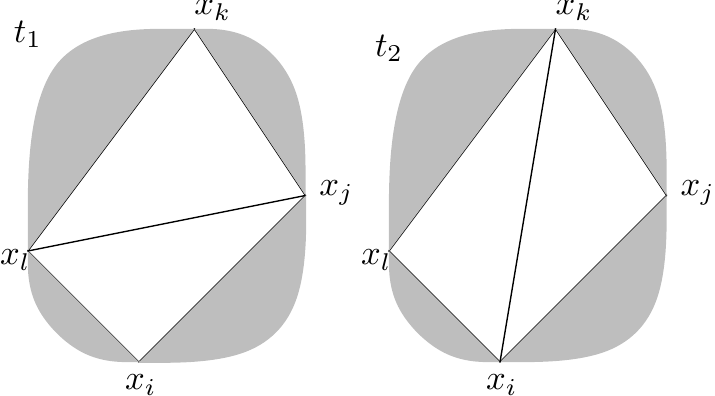} \]
\caption{The shaded regions are triangulated regions which are the same in $t_1$ and $t_2$.  Consider $s_{t_1}$ and $s_{t_2}$ at coordinate $i$.  In $s_{t_1}$ this coordinate is given by the area of those shaded triangles incident to $x_i$ (denoted by $b_i$ in the proof) along with $a(i,j,l)$.  In $s_{t_2}$ this coordinate is given by $b_i$ along with $a(i,l,k) + a(i,j,k) = a$.}\label{triangulate}
\end{figure}

Returning to $\tau(W)$ and $\tau(X)$, since the direction vector of an edge in $A_n$ depends only on its quadrilateral and the bijection between the edges bounding $\tau(W)$ and the edges bounding $\tau(X)$ preserves the quadrilateral, we have the same direction vectors bounding $\tau(W)$ as $\tau(X)$.  Furthermore, the bijection also maintains the relative configuration of the edges and so the edges corresponding to $\tau(W)$ and $\tau(X)$ are parallel.
\end{proof}

In order to prove the converse of Proposition~\ref{prop easy way}, we first require a preliminary lemma.

We also need the following notation: given a polygon dissection $D$ of an $n$-gon, let $V_{>3}(D)$ be the set of vertices of the polygon which are on at least one non-triangle face of $D$.  (Note that $D$ may have many non-triangle faces and so $V_{>3}(D)$ will be the union of the vertices of each non-triangle face.) For the purposes of this proof, we also write $[x_i,x_j]$ for the set of vertices from $x_i$ to $x_j$ including the endpoints, and $(x_i,x_j)$ for the same set but excluding the endpoints. 

\begin{lem}\label{lem good quads}
Let $D_1$ and $D_2$ be two dissections of a convex $n$-gon defined by the vertices $x_1, \ldots, x_n$.  Suppose that $V_{>3}(D_1) = V_{>3}(D_2)$, but that the $D_1$ and $D_2$ do not have the same set of non-triangle faces.  Then there are four vertices $x_i$, $x_j$, $x_k$, $x_l$ satisfying the following conditions:\begin{enumerate} \item in either $D_1$ or $D_2$ all four vertices are on the same non-triangle face. \item in the other dissection, there is a diagonal of the dissection such that exactly one of $\{x_i, x_j, x_k, x_l\}$ lies on one side of the diagonal (so the others are either end points of this diagonal or strictly on the other side). \end{enumerate}
\end{lem} 

The situation described in the lemma is illustrated in Figure~\ref{configs}.

\begin{figure}[t]
\[\includegraphics{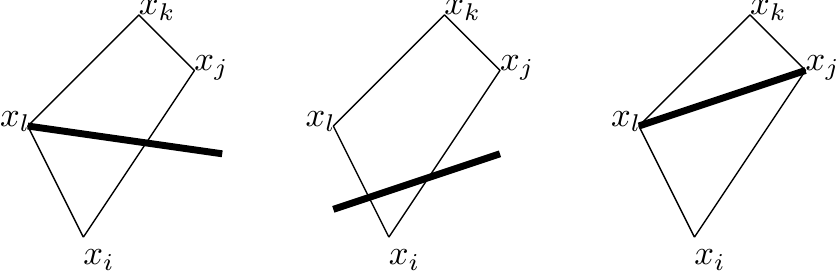} \]
\caption{Up to reflection and to rotation of $x_i$, $x_j$, $x_k$, $x_l$, the possible configurations of the quadrilateral from one dissection and diagonal from the other dissection (thick edge) illustrated together on the same vertex set.}\label{configs}
\end{figure}

\begin{proof}
  Since $V_{>3}(D_1) = V_{>3}(D_2)$, the vertices of the non-triangle faces of $D_1$ and $D_2$ are the same, but $D_1$ and $D_2$ do not have the same set of non-triangle faces.
  Suppose that no diagonal of $D_1$ crosses through a non-triangle face of $D_2$.  Then all non-triangle faces of $D_2$ must lie strictly within faces of $D_1$.  Since the non-triangle faces of the two dissections are different at least one of these containments must be proper, and so some diagonal bounding one of the non-triangle faces of $D_2$ must cross through a non-triangle face of $D_2$.  Therefore, swapping $D_1$ and $D_2$ if necessary,  
  %Since $D_1$ and $D_2$ do not have the same set of non-triangle faces, but the vertices of the non-triangles faces are the same, $V_{>3}(D_1) = V_{>3}(D_2)$,
  there must be some diagonal $d$ of $D_1$ that crosses through a non-triangle face $f$ of $D_2$.
  If $f$ has degree at least 5 then it must contain four vertices in the configuration described in the statement.

Suppose no $f$ of degree at least 5 crossed by a diagonal of the other dissection occurs in either $D_1$ or $D_2$.  Then all the non-triangle faces of $D_2$ which are crossed by a diagonal of $D_1$ must be quadrilaterals and likewise for $D_1$.  This means that any non-triangle and non-quadrilateral faces must be common between $D_1$ and $D_2$.  These common faces are irrelevant for the statement as none of them can be the face containing $x_i, x_j, x_k, x_l$, and so by triangulating any common non-triangle faces, we may assume that $D_1$ and $D_2$ have only quadrilateral and triangle faces.

Assume for a contradiction that no quadrilateral as in the statement exists, and let $d$ continue to be a diagonal in $D_1$ crossing a non-triangle face $f$ of $D_2$. Then $f$ must have two vertices strictly on each side of $d$. We name the vertices on one side of $d$ as $x_i$ and $x_j$, and those on the other side as $x_k$, $x_l$ in that order; see Figure~\ref{inductive_step}.  Let $g$ be the face in $D_1$ that is incident to $d$ on the $x_i, x_j$ side of $d$.

We first show that none of the vertices defining $g$ can lie in the cyclic interval $[x_i,x_j]$. Indeed, if $g$ is a triangle then we immediately see that its third vertex cannot lie in the interval $[x_i,x_j]$, as we could obtain a configuration as described in the statement by replacing $d$ with one of the other edges of $g$.

Similarly, if $g$ is a quadrilateral then by the same argument none of its vertices can lie in the interval $(x_i, x_j)$, nor can we have exactly one of $x_i$ or $x_j$ being a vertex of $g$ for the same reason. If $x_i$ and $x_j$ are both vertices of $g$ (so $g$ is the quadrilateral defined by $x_i$, $x_j$, and the two endpoints of $d$) then we still obtain a configuration as described in the lemma statement by swapping the roles of $f$ and $g$: now the four vertices would be the ones defining $g$, and the diagonal would be one of the edges of $f$. (For instance $d$ can now be the edge connecting $x_l$ to $x_i$ or $x_j$ to $x_k$.) Thus we can conclude that none of the vertices of the face $g$ are in the interval $[x_i, x_j]$.  

There are now two possibilities for $g$: either it has one vertex in the interval $(x_j, x_k)$ and the rest in $(x_l, x_i)$ (or vice versa), or it has two vertices in $(x_j, x_k)$ and two in $(x_l, x_i)$. In the first case, $g$ must be a triangle since if $g$ was a quadrilateral then its four vertices provide the desired configuration (with the edge of $f$ that connects $x_j$ and $x_k$ providing the desired diagonal). In the second case, $g$ has four vertices and is therefore a quadrilateral.

These cases are illustrated in Figure~\ref{inductive_step}. If $g$ is a triangle then let $d'$ be the other diagonal across $f$ (leftmost two cases in the figure), while if $g$ is a quadrilateral let $d'$ be the other diagonal that crosses $f$. In all cases, note that $d'$ is strictly closer to $[x_i,x_j]$ than $d$ was, since we chose $g$ to be the face incident to $d$ on the $x_i,x_j$ side.

Now $d'$ is also a diagonal of $D_1$ that crosses the non-triangle face $f$ of $D_2$, and we can repeat the above argument with $d'$ instead of $d$. This process can be repeated indefinitely but the polygon is finite, yielding a contradiction to our assumption that no quadrilateral with the configuration described in the statement exists.
\end{proof}

\begin{figure}
\[\includegraphics{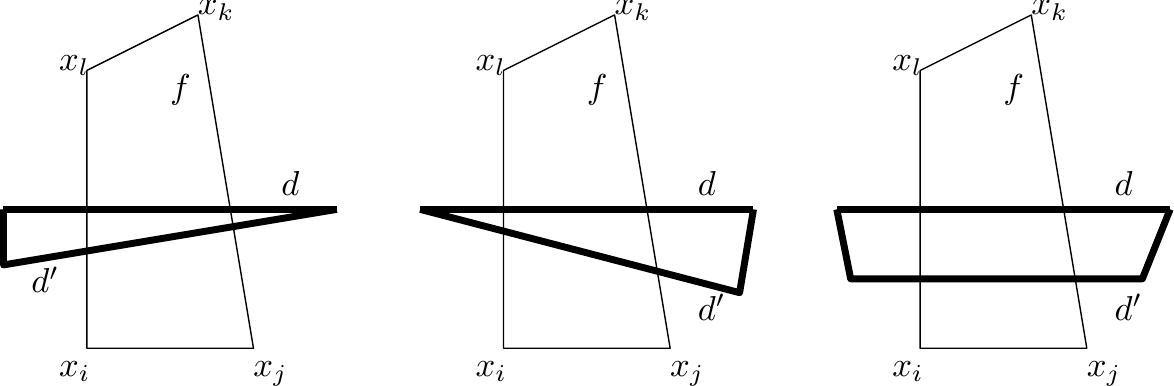} \]
\caption{The three possible configurations in the final part of the proof of Lemma~\ref{lem good quads}.  The thin lines come from the polygon dissection with $f$ as a face while the thick lines come from the polygon dissection with $d$ as an edge.}\label{inductive_step}
\end{figure}

\begin{prop}\label{prop hard way}
  If two weakly admissible Wilson loop diagrams are not equivalent then their faces in $A_n$ are not parallel.
\end{prop}

\begin{proof}
  Let $W$ and $X$ be inequivalent.  If $\tau(W)$ and $\tau(X)$ do not have the same number of diagonals (hence their faces in $A_n$ do not have the same dimension) then they are not parallel by definition. Now suppose that $\tau(W)$ and $\tau(X)$ do have the same number of diagonals and hence their faces in $A_n$ have the same dimension.  Let the face associated to $\tau(W)$ in $A_n$ be $f_W$ and the face associated to $\tau(X)$ be $f_X$.

As observed in the proof of Proposition~\ref{prop easy way}, any non-zero vector in $\R^n$ parallel to an edge of $A_n$ has non-zero entries only in the coordinates corresponding to the four vertices bounding the quadrilateral of the dissection associated with the edge.  Thus if $V_{>3}(\tau(W)) \neq V_{>3}(\tau(X))$ (recall the notation defined before Lemma \ref{lem good quads}) then the span of the vectors parallel to their edges has different support and so they cannot be parallel.  So we can assume that $V_{>3}(\tau(W)) = V_{>3}(\tau(X))$.

    Since $\tau(W)$ and $\tau(X)$ do not have the same set of non-triangle faces, there must be some diagonal $d$ of $\tau(W)$ that crosses through a non-triangle face $f$ of $\tau(X)$. 
    It will be most convenient to consider faces in terms of their normal vectors. We make use of the construction discussed on page 8 of \cite{CSZinequivalent}: namely, pick one side of $d$ to be positive and define the $i$th coordinate of the vector $\omega^+_d$ to be
    \[
    \omega^+_d(i) = \begin{cases} \text{dist}(x_i, d) & \text{if $x$ is on the positive side of $d$,}\\ 0 & \text{otherwise.} \end{cases}
    \]
    We have from \cite{CSZinequivalent} (see the proof of Proposition 3.5 therein) that $\omega^+_d$ is a normal vector for any face with $d$ as a diagonal.  However we can also check this fact in a completely elementary way using the shoelace formula for the area of a triangle.  The shoelace formula says that if $(x_1, y_1), (x_2, y_2), (x_3, y_3)$ are the coordinates of the corners of a triangle in counterclockwise order then the area of the triangle is \[\textstyle\frac{1}{2}(x_1y_2 - x_1y_3 + x_2y_3 - x_2y_1 + x_3y_1 - x_3y_2).\]

    We now check that $\omega^+_d$ is a normal vector for any face associated to a dissection including $d$, which we accomplish by showing it is normal to every edge bounding the face.  To do so, choose the coordinate system for the polygon so that $d$ is along the $x$-axis.  Let $(a_i, h_i)$, $i\in \{1,2,3,4\}$ be the coordinates of the vertices of a quadrilateral in counterclockwise order and which corresponds to an edge bounding the face.  Since the dissection includes $d$, either all four points are on the negative side of $d$ and so the dot product of $\omega^+_d$ with the direction vector for the quadrilateral is immediately $0$, or, by the formula in equation \eqref{eq:direction vector for quadrilateral} and the shoelace formula, the dot product of $\omega^+_d$ with the vector associated to the quadrilateral is
   \begin{align*}
       & \hspace{1cm}\textstyle\frac{1}{2}h_1(a_2h_3-a_2h_4+a_3h_4-a_3h_2+a_4h_2-a_4h_3) \\
       &\hspace{3cm} - \textstyle\frac{1}{2} h_2(a_1h_3-a_1h_4+a_3h_4-a_3h_1+a_4h_1-a_4h_3) \\
       &\hspace{5cm}+ \textstyle\frac{1}{2}h_3(a_1h_2-a_1h_4+a_2h_4-a_2h_1+a_4h_1-a_4h_2) \\
       &\hspace{7cm}- \textstyle\frac{1}{2}h_4(a_1h_2-a_1h_3+a_2h_3-a_2h_1+a_3h_1-a_3h_2), 
    \end{align*}
    which also evaluates to 0.

    Now, $\tau(W)$ and $\tau(X)$ satisfy the hypotheses of Lemma~\ref{lem good quads}, so we can assume that $f$ and $d$ have the property that we can choose a side of $d$ to be the positive side in such a way that exactly one vertex of $f$ is strictly on the positive side.  Then $\omega^+_d$ and the direction vector of $d$ have common support exactly on that one vertex and so their dot product is non-zero.  Since the direction vector of $d$ is contained in the direction of $\tau(X)$ but is not orthogonal to $\omega^+_d$, while $\omega^+_d$ is normal to the face of $\tau(W)$, it follows that the faces associated to $\tau(W)$ and $\tau(X)$ are not parallel, as desired.
\end{proof}

Propositions~\ref{prop easy way} and \ref{prop hard way} together imply the main result of this section:
\begin{thm}\label{thm:count inequiv diagrams}
  The number of inequivalent weakly admissible Wilson loop diagrams on $n$ is the number of non-parallel faces in the associahedron $A_n$.
\end{thm}

\section{Conclusion}

In summary, using matroids and polygon dissections we have counted Wilson Loop diagrams that are equivalent in the sense of having the same positroid, and characterized Wilson Loop diagrams that are inequivalent in terms of faces of the associahedron $A_n$ up to parallelism.  The second paper in this series calculates the dimension of the positroid cell associated to a Wilson loop diagram in a very concrete way by counting pluses in a Le diagram, and looks at the denominator of the integrand of the Wilson Loop diagram.  All of these results are interesting mathematics while also serving to better understand the physics of scattering amplitudes.

\bibliographystyle{abbrv}
\bibliography{WLDbib}

\end{document}